\documentclass[preprint,12pt]{elsarticle}

\usepackage{amssymb}
\usepackage{amsmath}
\usepackage{amsmath}    
\usepackage{amsthm}
\usepackage{amsmath,bm}
\usepackage{amssymb}
\usepackage{tikz}
\usepackage{caption}
\usepackage{subcaption}
\usepackage{multicol}
\usepackage{hyperref}
\usepackage{paralist}

\usetikzlibrary{arrows}
\usetikzlibrary{shapes}
\usetikzlibrary{automata}
\usetikzlibrary{positioning}
\usetikzlibrary{trees}

\usepackage{enumerate}

\usepackage{xspace} % for marcros

\newtheorem{definition}{Definition}
\newtheorem{example}{Example}
\newtheorem{theorem}{Theorem}
\newtheorem{claim}{Claim}
\newtheorem{proposition}{Proposition}
\newtheorem{lemma}{Lemma}
\newtheorem{corollary}{Corollary}

\makeatletter
\def\tagform@#1{\maketag@@@{(\ignorespaces\textit{Eq.~#1}\unskip\@@italiccorr)}}
\makeatother
\newcommand{\numatts}{k}

\newcommand{\myparagraph}[1]{\vspace{0.15cm}\noindent\textbf{#1.}}

\newcommand{\eat}[1]{}

\newcommand{\ptime}{\textsc{PTIME}\xspace}%michael: all the classes
\newcommand{\np}{\textsc{NP}\xspace}
\newcommand{\nlspace}{\textsc{NL}\xspace}

\newcommand{\pspace}{\textsc{PSPACE}\xspace}

\DeclareDocumentCommand\nodedata{o}{\dataof\IfValueTF{#1}{(#1)}{}\xspace}
\newcommand{\restrictor}{\mathfrak{t}\xspace}
\newcommand{\pattern}{\mathfrak{p}\xspace}
\newcommand{\primepattern}{\mathfrak{q}\xspace}

\newcommand{\subpattern}{\mathfrak{q}\xspace}
\newcommand{\subsubpattern}{\mathfrak{r}\xspace}

\newcommand{\fv}[1]{\textsc{FV}(#1)}
\newcommand{\semantics}[2]{[\![{#1}]\!]_{#2}}

\newcommand{\outm}{\textsc{OUT}}
\newcommand{\inm}{\textsc{IN}}
\newcommand{\pathsep}{\thinspace}% separate items in paths

\newcommand{\shortest}{\textsc{shortest} \xspace}
\newcommand{\simple}{\textsc{simple} \xspace}
\newcommand{\nonsimple}{\textsc{nonsimple} \xspace}

\newcommand{\visitall}{\textsc{visitall} \xspace}
\newcommand{\idm}{\textsc{id} \xspace}
\newcommand{\dataof}{\textsc{dataof} \xspace}
\newcommand{\propm}{\textsc{prop} \xspace}
\newcommand{\datam}{\textsc{data} \xspace}
\newcommand{\difm}{\textsc{diff} \xspace}

\newcommand{\lenm}{\textsc{len} \xspace}
\newcommand{\shortestsimple}{\textsc{shortestsimple} \xspace}
\newcommand{\dataconnection}{\textsc{DataConnection} \xspace}
\newcommand{\datalink}{\textsc{DataLink} \xspace}

\newcommand{\edn}{FO(ERDPQ)\xspace}
\newcommand{\foerd}{\edn}

\newcommand{\en}{ECRPQ$_\neg$\xspace}
\newcommand{\ednt}{FO$^*$(ERDPQ)\xspace}

\newcommand{\uedn}{$\forall$(ERDPQ)\xspace}

\newcommand{\cn}{CRPQ$_\neg$\xspace}
\newcommand{\cdn}{FO(RDPQ)\xspace}
\newcommand{\ford}{\cdn}

\newcommand{\crp}{CRPQ\xspace}
\newcommand{\crd}{CRDPQ\xspace}
\newcommand{\ep}{ECRPQ\xspace}

\newcommand{\rd}{RDPQ\xspace}

\newcommand{\rp}{RPQ\xspace}

\newcommand{\updatet}{\textbf{Upd}}
\newcommand{\precondt}{\textbf{Pre}}

\newcommand{\eqid}{\equiv_{\idm}\xspace}
\newcommand{\eqdata}{\equiv_{\datam}\xspace}
\newcommand{\eqdatai}{\equiv_{\datam,i}\xspace}

\newcommand{\gpc}{GPC\xspace}

\newcommand{\wl}{WL\xspace}

\newcommand{\threeval}{\textbf{3val}}
\newcommand{\repeatm}{\textbf{repeat}}
\newcommand{\first}{\textbf{first}}
\newcommand{\last}{\textbf{last}}

\newcommand{\visit}{\textbf{visit}}
\newcommand{\beginm}{\textbf{begin}}
\newcommand{\even}{\textbf{even}}
\newcommand{\tran}{\textbf{tran}}

\newcommand{\fotc}{FO$^*$\xspace}
\newcommand{\fotcd}{FO$^*(\eqdata)$\xspace}
\newcommand{\fod}{FO$(\eqdata)$\xspace}

\newcommand{\fotcr}{(FO$^*_{3}$)$(\eqdata)$\xspace}
\newcommand{\fo}{FO\xspace}

\newcommand{\gpcrd}{GXPath$_{\textsc{reg}}(\eqdata)$\xspace}

\newcommand{\ra}{RDPA\xspace}
\newcommand{\nullnode}{\sharp\xspace}

\newcommand{\lb}[1]{\mathbf{lb}(#1)}
\newcommand{\tdp}[1]{\mathbf{dp}(#1)}
\newcommand{\mwl}{MWL\xspace} 
\usepackage[pagewise]{lineno}
\journal{Journal of Computing and Systems Sciences}

\begin{document}

\begin{frontmatter}



\title{Revisiting the Expressiveness Landscape of Data Graph Queries}


%\author{Michael Benedikt and Anthony Widjaja Lin and Di-De Yen}

%\affiliation{organization={},%Department and Organization
 %           addressline={}, 
  %          city={},
   %         postcode={}, 
    %       state={},
     %       country={}}

\author[1]{Michael Benedikt}
\author[2,3]{Anthony Widjaja Lin}
\author[4]{Di-De Yen}

\affiliation[1]{organization={Department of Computer Science,
University of Oxford},%Department and Organization
            %addressline={}, 
            %city={},
            %postcode={}, 
            %state={},
            country={United Kingdom}}
\affiliation[2]{organization={University of Kaiserslautern-Landau},%Department and Organization
            %addressline={}, 
            %city={},
            %postcode={}, 
            %state={},
            country={Germany}}
\affiliation[3]{organization={MPI-SWS},%Department and Organization
            %addressline={}, 
            %city={},
            %postcode={}, 
            %state={},
            country={Germany}}
\affiliation[4]{organization={Department of Computer Science, University of Liverpool},%Department and Organization
            %addressline={}, 
            %city={},
            %postcode={}, 
            %state={},
            country={United Kingdom}}

\begin{abstract}
The study of graph queries in database theory has spanned more than three decades, resulting in a multitude of proposals for graph query languages.  We can identify three main families of languages, with the canonical representatives being: (1) regular path queries, (2) walk logic, and (3) first-order logic with transitive closure operators. This paper provides a complete picture of the expressive power of these languages in the context of data graphs. Specifically, we consider a graph data model that supports querying over both data and topology. For example, ``Does there exist a path between two different persons in a social network with the same last name?''. We also show that an extension of (1) with regular path comparisons, augmented with transitive closure operators, can unify the expressivity of (1)--(3).
\end{abstract}



%% Keywords

%\keywords{graph databases, query languages, data graphs, expressiveness} 

\end{frontmatter}

\section{Introduction}

Graph databases are data models with a multitude of natural applications including social networks, semantic web, biology, ecology, supply chain management, and business process modeling. With graphs as one of the main repositories of data, \emph{graph querying} has become a major component in data wrangling. Indeed, graph database systems -- like Neo4j, Oracle, TigerGraph, among many others -- have increasingly found usage in a plethora of application domains. 

One important feature in graph query languages is the ability to query for \emph{paths} over a given database. 
In fact, much of the effort in the study of graph query languages in database theory has been motivated by the need to support path queries. For example, in real world query languages (like Neo4j Cypher queries) one uses  \emph{patterns} like:
\begin{verbatim}
    (?x)=[:friend*]=>(Erdos)
\end{verbatim}
to find nodes that are connected to Erdos by the friend relation (assuming it is symmetric). Similar capabilities are also supported by SPARQL property paths. A systematic study of query languages supporting ``path patterns'' has been undertaken in database theory since the 1990s, resulting in a plethora of query languages over graphs. The core of such languages can be classified into three categories: (1) extensions of regular path queries \cite{MW95,LMV16,Lin:12:DS}, (2) walk logic \cite{Hellings:13:ICDT,BFL15}, and (3) first-order logic with transitive closure operators \cite{immermanbook,LMV16}. 
In addition, while the original query languages were studied within the basic graph database setting of edge-labelled graphs with finitely many labels, recent developments suggest the importance of supporting \emph{data} in the model, not just \emph{topology}. For example, if a node represents a person (with data including \texttt{age}, \texttt{firstname}, \texttt{lastname}, etc.) in a social network and an edge represents a ``friend-of'' relation, then we might be interested in a pair of friend-of-friends (i.e. transitive closure of friend relations) with the same data component \texttt{lastname}. To support such queries, an extended graph data model of \emph{data graphs} \cite{LV12,LMV16} was proposed, wherein  additional binary relations $\equiv_{\datam,i} \subseteq V \times V$ for $1 \leq i \leq n$ over nodes in the graph check whether two nodes have the same $i^{th}$ data component (e.g. \texttt{lastname}). The bulk of graph queries can be (and have been) extended to the setting of data graphs \cite{LV12,LMV16,BFL15,FJL22}.

%Owing to the novelty of property graphs and GQL, query languages over property graphs are still very much evolving and, moreover, many language features that have been studied over data graphs (e.g. those in walk logic) have never been studied in the context of property graphs. 
%This poses a certain challenge for a systematic study of expressive power of graph query languages.
%\footnote{For example, some querying capabilities that are allowed in Neo4j are not yet covered in existing query languages over property graphs.}. 
%For this reason, we focus in this paper on the more developed data model of data graphs. 
%\michael{I agree with the reviewer that we need to give another reason. You will have to remind me the various differences
%with property graphs and our model}

\myparagraph{Zoo of query languages over data graphs} 
We proceed by first surveying the categories (1)--(3) of graph query languages. The first category of query language is based on the idea of using regular languages to describe ``path patterns''. This originates in 
%michael: do not use "featured in" 
the seminal paper of Mendelzon and Wood \cite{MW95}, which introduces the so-called \emph{Regular Path Queries (\rp)}. For example, to describe the friend-of-friends relation, one may simply write the path pattern $(\texttt{friend})^*$, where $\texttt{friend}$ is a name of a relation in the database describing the friend-of relation. The class \rp has been extended  to operate over data graphs \cite{LV12,LMV16}, and to support unions/conjunctions \cite{CGLV00,DT01,FLS98}. For example, the query language Conjunctive Regular Path queries denoted \crp arose as part of the language GraphLog \cite{graphlog}.
It has also been extended to support regular path comparisons  \cite{Lin:12:DS}, utilized in the language Extended \crp (\ep). In particular, register automata -- or, equivalently, regular expressions with memory --  are used to extend the notion of regular path patterns to \emph{data paths}, an alternating sequence of data and edge labels. In register automata an unbounded register/memory is required to store data. For example, to enforce that the start node $v$ and end node $w$ correspond to persons with the same \texttt{lastname} attribute, an automaton can save the last name of the first person $v$ in the register and then check if it is the same as the last name of the last person $w$. The idea of using register automata to express regular patterns over data paths, or tuples of data paths,  can also be easily extended to \crp and \ep, as well as their extensions with negations \cn, and \en. To emphasize the data model used by the query languages, on which they operate, we denote \rp, \crp, etc. by \rd, \crd, etc., as suggested by \cite{LMV16}.

The second paradigm is that of {\em Walk Logic} (\wl) \cite{Hellings:13:ICDT}, which takes paths within graphs as the fundamental data item referenced within query variables. Paths are manipulated through first-order predicate logic operations: roughly the paradigm is ``relational calculus for paths''. As explained in \cite{BFL15}, \wl can be construed as a query language over data graphs.

The third paradigm is inspired by first-order logic with transitive closure operators \cite{immermanbook}. This logic was studied in the setting of data graphs in \cite{LMV16}, where additional binary relations $\eqdatai \subseteq V \times V$ over nodes in the graph are introduced that check whether two nodes have the same $i^{th}$ data component. This resulting logic (called {\fotcd}) subsumes GXPath \cite{LMV16} and also regular (data) queries \cite{RRV17}.

Finally, we mention another language that can be related to the first paradigm, the recently proposed \emph{Graph Pattern Calculus} (\gpc)  \cite{GPC:22:Francis, gqlinvitedicdt}, which arose during the ongoing standardization efforts \cite{gql} of SQL/PGQ and GQL for the more expressive graph data model of property graphs.

\begin{figure}[t]
\centering
\scalebox{0.6}{
\begin{tikzpicture}
\centering
[node distance=160pt]

\tikzset{invisible/.style={minimum width=0mm,inner sep=0mm,outer sep=0mm}}

\node (v4) at (-3.5,3) {\edn };
\node (v6) at (-3.5,1) {\gpc};
\node (v8) at (3,3) {\fotcd};
\node (v9) at (3,1) {\gpcrd};
\node (v10) at (-1.75,2) {\cdn};
\node (v11) at (-0.5,1) {\rd};
\node (v14) at (-6.5, 1) {\wl};

%p
\path
(v4)
   edge [<-] node[above, rotate=-90] {} (v6) 

(v8)
    edge [<-,red,thick] node[above, rotate=-90] {} (v9)
(v8)
    edge [<-] node[above] {} (v11)
(v4)
    edge [<-,red,thick] node[above] {} (v10) 
(v10)
    edge [<-,red,thick] node[above, rotate=-90] {\textcolor{red}{\ }} (v11)

(v14)
  edge [->] node[above, rotate=13] {} (v4) 

;

\end{tikzpicture}
}
\caption{Prior query languages (extended with data). In the diagrams, for any pair of languages $L$ and $M$, the arrow $L \rightarrow M$ signifies that $M$ is more expressive than $L$. Languages $L$ and $M$ are considered incomparable if there is no directed path (after taking transitive closure) between them.}
\label{fig:existing}
\end{figure}

\smallskip
\noindent
\textbf{Contributions:}
We begin by studying the prior proposals for graph query languages in the framework of data graphs. We analyze the expressiveness of these languages, identifying which containments hold: see Figure~\ref{fig:existing}. In the figure, the containment results indicated by \emph{red} edges  are either known or trivial, while other containment and incomparability results are new. Specifically, we show that over data graphs  the most expressive existing language \edn from the first paradigm strictly subsumes \wl, \rd, as well as \gpc. Although we show that {\fotcd}, from Category (3), subsumes \rd, we prove that it is incomparable with \edn, as well as \gpc. This leaves the question of whether there is a natural way to reconcile {\edn} and {\fotcd}.
Towards resolving this, we propose a language that subsumes all the existing ones. A summary of the revised landscape in terms of expressive power can be found in Figure~\ref{fig:summary}.

\begin{figure}[t]
\centering
\scalebox{0.6}{
\begin{tikzpicture}
\centering
[node distance=160pt]

\tikzset{invisible/.style={minimum width=0mm,inner sep=0mm,outer sep=0mm}}

\node (v1) at (0,5) {\ednt};
\node (v4) at (-3.5,3) {\edn};
\node (v6) at (-3.5,1) {\gpc};
\node (v8) at (3,3) {\fotcd};
\node (v9) at (3,1) {\gpcrd};
\node (v10) at (-1.75,2) {\cdn};
\node (v11) at (-0.5,1) {\rd};
\node (v14) at (-6.5, 1) {\wl};

%p
\path
(v1)
    edge [<-] node[above] {} (v4)
(v1)
    edge [<-] node[above] {} (v8)  
(v4)
   edge [<-] node[above, rotate=-90] {} (v6) 

(v8)
    edge [<-] node[above, rotate=-90] {} (v9)
(v8)
    edge [<-] node[above] {} (v11)
(v4)
    edge [<-] node[above] {} (v10) 
(v10)
    edge [<-] node[above, rotate=-90] {\textcolor{blue}{\ }} (v11)

(v14)
  edge [->] node[above, rotate=13] {} (v4) 

;

\end{tikzpicture}
}
\caption{Expressiveness of languages.}
\label{fig:summary}
\end{figure}

At a high level, our proofs utilize the usual expressiveness toolbox in  finite model theory   -- e.g. automata and complexity-theoretic arguments. But their application to some of the finer-grain comparisons we consider here requires some subtlety. For example, since \edn combines regular expression power with some arithmetic, in separating \edn from {\mwl} -- an extension of \wl introduced in this paper with expressiveness between \edn and {\wl} -- we require a combination of circuit complexity bounds and rewriting techniques; in showing \rd is not subsumed by \gpc, we need to bound the expressiveness of \gpc, and thus require a variation of the pumping lemma for automata that is tailored towards the fine points of the \gpc syntax.

\myparagraph{Limitations: Deviations from other proposals and standards}
This paper is a step towards a better understanding of the expressiveness landscape for these languages. In doing so, we have made a number of simplifications both at the data model and query language level, relative to languages considered by both practitioner query languages and the evolving standards.
One of these simplifications is at the data model level.
%As is the case for typical work in graph database theory to date, edge-labelled graphs and data graphs are idealized data models that do not fully capture the complexity of data models that are adopted in real-world languages like Cypher or SPARQL property paths.
There is a more flexible (and complex) data model called \emph{property graphs} \cite{foundation-survey,gql}, which allows both edge and node labels, both of which can be varying or bound by a specific ``schema''. This model has been in the spotlight in database theory in recent years, and property graphs will serve as the data model for GQL, the standard graph query language that is currently under intense development (e.g. see \cite{gql,GPC:22:Francis, gqlinvitedicdt}).  

Another simplification involves directionality in navigating a graph. We deal with a directed graph model, but in our formalization of the regular path query and walk logic paradigms, we allow only navigation in the forward direction. This is consistent with formalization of traditional query languages like R(D)PQ \cite{LMV16}, ECRPQ \cite{Lin:12:DS}, and Walk Logic \cite{Hellings:13:ICDT}. On the other hand, queries in $\gpc$ allow navigation in both directions. Such a small ``mismatch'' would immediately render $\gpc$ to be incomparable with the aforementioned languages for trivial reasons. For a fair comparison, one could either (i) extend RDPQ, ECRPQ, and Walk Logic with bidirectional navigations, or (ii) allow only a forward direction in $\gpc$. We opted for (ii). See further discussion in Section \ref{sec:conc}.

A third deviation concerns support for \emph{trails}: one of our considered languages, $\gpc$, allows one to restrict to paths with no repetition among nodes. We omit this feature, but only to keep the exposition briefer: we show in Section \ref{sec:conc} that all the results still hold when the syntax of $\gpc$ is extended to include this.
Although our work assumes a simpler data and query model, we believe our results will be useful in considering the expressiveness of languages making use of the full property graph data model.

\section{Preliminaries}\label{sec:preliminaries}

\noindent
\myparagraph{General notation} %use a consistent format for paragraph
We use $\mathbb{N}$ and $\mathbb{Z}_{>0}$ to denote the sets of non-negative integers and positive integers, respectively. For any two $i$ and $j$ in $\mathbb{N}$ where $i<j$, the notation $[i,j]$ represents the set $\{i, i+1, \dots, j\}$. If $i=1$, we simply write $[j]$, abbreviating $[1,j]$. Let $m$ be a mapping, and let $d$ and $i$ be elements in the domain and image of $m$, respectively. The notation $m[d\mapsto i]$ represents the mapping equivalent to $m$, except that $d$ is now mapped to $i$.

\begin{figure}[!h]
\centering
\includegraphics[width=0.95\textwidth]{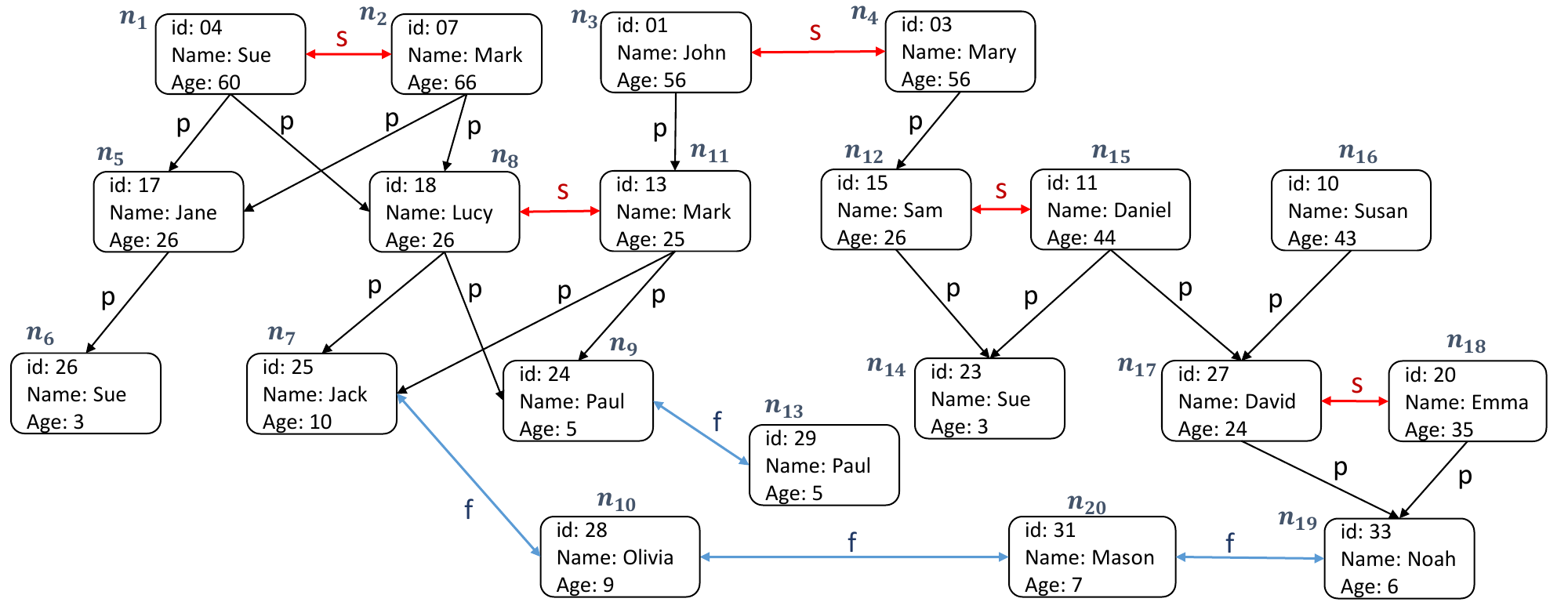}
\caption{Data graph $G$, where p stands for parent, s for spouse, and f stands for friend.}
\label{fig:DataGraph_ex}
\end{figure}

\myparagraph{Data model} %use a consistent format for paragraph
In this work, we are concerned with graphs that incorporate data, referred to as {\em data graphs}.  Specifically, each node corresponds to a unique id and is accompanied by additional attributes.
For instance, as shown in Figure~\ref{fig:DataGraph_ex}, nodes possess attributes such as Name (with the data type string) and Age (with the data type natural number). To denote the sets of ids and other pertinent attributes (``properties''), we use $\mathcal{D}_{\idm}$ and $\mathcal{D}_{\propm}$, respectively. Here, $\mathcal{D}_{\propm}$ represents a $k$-ary relation, where $\numatts$ is the number of attributes that nodes have, excluding the id. For the example shown in Figure~\ref{fig:DataGraph_ex}, $k$ is 2.
The data domain, denoted as $\mathcal{D}$, is the Cartesian product of $\mathcal{D}_{\idm}$ and $\mathcal{D}_{\propm}$ when $k > 0$; otherwise, when $k = 0$, $\mathcal{D}$ is equal to $\mathcal{D}_{\idm}$.

Formally, a data graph is a directed graph in which each node corresponds to a data value. To define the data value of a node, we introduce two key functions, $\idm$ and $\nodedata$. 
The function $\idm$ maps each node $v$ to its unique id. For example, in Figure~\ref{fig:DataGraph_ex}, we have $\idm(n_2) = 07$, $\idm(n_9) = 24$, and $\idm(n_{13}) = 29$. When $k > 0$, for each $i \in [k]$, we define $\nodedata_i$ as the function that maps each node $v$ to the $i^{\text{th}}$ component of $\nodedata[v]$. In other words, if $\nodedata[v] = (d_1, \dots, d_i, \dots, d_k)$, then $\nodedata_i(v) = d_i$. For instance, for nodes $n_2$, $n_9$, and $n_{13}$, we have $\nodedata(n_9) = (\textit{Paul}, 5) = \nodedata(n_{13})$, and $\nodedata_2(n_2) = 66$. 
Since the function $\idm$ is injective, for any two nodes $v$ and $v'$ in $V$, $\idm(v) = \idm(v')$ if and only if $v = v'$. The tuple $(\idm(v), \nodedata(v)) \in \mathcal{D}$ represents the {\em data value} of node $v$.
For simplicity, we often treat a node's id as synonymous with the node itself.

Our definition of data graphs extends the one provided in \cite{LV12}. The primary distinction is that, in our definition, each node has a unique id as part of its data value, as is done in \cite{BFL15}. This slight modification allows us to model node identity in all our query languages without cluttering them with additional syntax. More importantly, treating ids as part of node data values is realized through the $\idm$ function, and this feature plays a significant role in the expressiveness of the \emph{regular data path query language} (\rd) and its extensions, defined in Section~\ref{sec:regular_based_languages}. 
For example, the query ``Does a node appear more than once along a path $\pi$?'' would be inexpressible in \rd and its extensions if ids were not considered part of node data values, or if the $\idm$ function were removed from our definition. In such a case, these languages and other languages considered in this paper would become incomparable. However, this would contradict one of our goals: to characterize existing query languages for graphs with data and identify their core features.

The formal definition of data graphs is as follows:

\begin{definition}[Data Graph]\label{def:data_graph}
A data graph $G$ over alphabet $\Sigma$ and data domain $\mathcal{D}$ is a quadruple $(V,E,\idm,\nodedata)$, where
\begin{itemize}
\item $V$ is a non-empty finite set of nodes.
\item $E \subseteq V\times \Sigma \times V$ is a set of labeled edges.
\item $\idm: V \rightarrow \mathcal{D}_{\idm}$ is an injective function mapping each node to an id.
\item $\nodedata: V \rightarrow \mathcal{D}_{\propm}$ is a function assigning an element in $\mathcal{D}_{\propm}$ to each node in $V$.
\end{itemize}
\end{definition}

A {\em path} $\rho$ in the data graph $G$ is a non-empty alternating sequence 
\[v_0 \pathsep a_1 \pathsep v_1  \pathsep \dots v_{n-1} \pathsep a_n \pathsep v_n
\]
where $v_i \in V$,
$a_j \in \Sigma$, and $(v_{l-1},a_l,v_l)$ forms an edge in $E$ for all $i$, $j$, and $l$. The length of $\rho$, denoted as $|\rho|$, is $n$. The set of {\em positions} of $\rho$ is $[0,n]$, with $v_i$ being the node at position $i$ for every $i \in [0,n]$. The {\em data path} corresponding to $\rho$ is the string $d_0 \pathsep a_1 \pathsep d_1 \dots a_n \pathsep d_n$, where $d_i$ represents the data value associated with $v_i$ for every $i \in [0,n]$. This data path is denoted as $\tdp{\rho}$. The {\em label} of $\rho$ is $\mathbf{lb}(\rho) = a_1\dots a_n$.
A path $\rho$ is considered a {\em simple path} if $v_i \neq v_j$ for every $i,j \in [0,n]$ and $i \neq j$.
For example, consider the  path in the graph from Figure \ref{fig:DataGraph_ex}:
\[
\rho = n_4 \pathsep {s} \pathsep n_3 \pathsep {p} \pathsep n_{11} \pathsep {p} \pathsep n_9 \pathsep {f} \pathsep n_{13}
\] 
It is a simple path, with $\lb{\rho} = {s \pathsep p \pathsep p \pathsep f}$ and $\tdp{\rho} = (03,\textit{Mary},56) \ldots\linebreak (29,\textit{Paul},5)$.

Given two paths $\rho = v_0  ~ a_1 ~ v_1 \dots v_n$ and $\rho' = v'_0 ~ a'_1 ~ v'_1 \dots v'_m$, $\rho$ can be concatenated with $\rho'$ if $v_n = v'_0$. The concatenated path, denoted as $\rho \cdot \rho'$, is given by $v_0 ~ a_1 \dots v_n ~ a'_1 ~  v'_1 \dots v'_m$. 
A graph $G$ is classified as a {\em chain} if it is connected and each node has at most one predecessor and at most one successor.

\myparagraph{Queries}
Informally, a {\em query} is a function that returns a Boolean value, taking as input a data graph over a given $\Sigma, \mathcal{D}$, along with a valuation that maps a predetermined set of free variables of different types (paths, nodes, positions) to values of the appropriate kind. While these languages often vary in the types of free variables they accommodate, all our query languages support {\em Boolean graph queries}, where the input is solely a data graph over $\Sigma, \mathcal{D}$. Usually, $\Sigma, \mathcal{D}$ will be evident from the context, so we'll omit them. 

\begin{definition}{Language subsumption} \label{def:subsumes}
A query language $L$ defines a set of queries, and we say language $L$ is {\em subsumed by} language $L'$ if every Boolean graph query expressible in $L$ is also expressible in $L'$.
\end{definition}

In previous discussions, we've addressed languages semantically as collections of queries, yet of course, languages of interest are usually defined using syntax for expressions. Often, we identify the expressions with their corresponding queries. Therefore, given a (syntactically defined) query language, a related problem is to evaluate an expression in the language on a given data graph.  This is the {\em query evaluation problem} for the language.

%\michael{Mention that we are fixing the ``schema'', which can be thought of as consisting of $\Sigma$ and ....}
%\dd{I don’t fully follow. We have defined the data model, data graphs, which are always of the form $(V,E,ID,DATAOF)$. Doesn’t this already imply that the schema is fixed, once $\Sigma$ and $\mathcal{D}$ are given?}
%michael: I think it is already addressed above

\section{Existing query languages for graphs with data}\label{sec:queries}

\subsection{Graph querying based on path quantification}
{\em Walk logic} (\wl) is a first-order graph query language that quantifies over variables ranging paths and positions within paths \cite{Hellings:13:ICDT}. With the capability to manipulate paths and positions, the query $Q_H$: ``Is there a Hamiltonian path in the graph?'' can be readily expressed using \wl.
Before providing an exact definition, it is instructive to understand how \wl expresses that a graph is Hamiltonian. A graph is Hamiltonian if it possesses a path that (1) is simple (with no repeated nodes) and (2) visits all nodes. This can be expressed using a logic that quantifies over paths and positions within a given path. We employ variables $\pi$ and $\omega$ for variables ranging over paths, and $\ell^\pi$ and $m^\pi$ variables ranging over positions within the path referenced by variable $\pi$.
Intuitively, each path variable $\pi$ in \wl is interpreted as a path $\rho = v_0 \pathsep a_1 \dots a_i \pathsep v_i$ within the data graph $G$. Each position variable $\ell^\pi$ represents the position $r \in [0, i]$ of a node along the path $\rho$. Recall that $v_r$ is the node at the \underline{position} $r$.
Hence, the following \wl formula captures that path $\pi$ is simple:
\[\phi_{\simple}(\pi):= \forall \ell^{\pi},m^{\pi}.\ \ell^{\pi} \neq m^{\pi} \rightarrow
\neg(\ell^{\pi} \eqid  m^{\pi}),\]
where $\ell^{\pi} \eqid m^{\pi}$ indicates that positions referenced by $\ell^{\pi}$ and $m^{\pi}$ point to the same node. 
A path $\rho$ visits all nodes in a graph $G$ if and only if for every path $\rho'$ and a node $n'$ in $\rho'$, there exists a node $n$ in $\rho$ such that $n$ is identical to $n'$. Therefore, the subsequent \wl formula verifies whether (2) is fulfilled:
\[\phi_{\visitall}(\pi):= \forall \omega.\ \forall m^\omega.\ \exists \ell^\pi.\ \ell \eqid m.\]
Hence, formula $\exists \pi.\ \phi_{\simple}(\pi) \wedge \phi_{\visitall}(\pi)$ expresses Hamiltonicity.

The syntax of \wl assumes a countably infinite set $\Pi$ of path variables, coupled with a set $\Lambda$ of position variables indexed by path variables. A position variable $\ell$ is of sort $\pi$ and is usually denoted as $\ell^\pi$ \cite{Hellings:13:ICDT,BFL15}, representing positions within the path $\pi$. The $\pi$ superscript is omitted when the context is clear.

\begin{definition}
Let $\Sigma$ be a finite alphabet, $\Pi$ a countably infinite set of path variables, and $\Lambda$ a set of position variables indexed by path variables. The formulas of \wl over $\Sigma$, $\Pi$, and $\Lambda$ are defined inductively as follows:
\begin{itemize}
\item $E_a(\ell^\pi,m^\pi)$, $\ell^\pi < m^\pi$, $\ell^\pi \eqid n^\omega$, and $\ell^\pi \eqdata n^\omega$ are atomic formulas, where $\pi, \omega \in \Pi$, $\ell^\pi, m^\pi, n^\omega \in \Lambda$, and $a \in \Sigma$.
\item If $\phi$ and $\psi$ are formulas, so are $\neg \phi$, $\phi \vee \psi$, $\exists \ell^\pi.\ \phi$, and $\exists \pi.\ \phi$.
\end{itemize}
Note: in \wl we can say $\ell^\pi <n^\omega$ only if $\pi = \omega$.
\end{definition}

Let $\phi$ be a \wl formula and $\alpha$ be a mapping that assigns paths in $G$ to path variables and node positions to position variables.
Assume $\alpha(\pi)=v_0 \pathsep a_1\dots a_i \pathsep v_i$, $\alpha(\omega)=w_0 \pathsep b_1\dots b_{j} \pathsep w_{j}$, $\alpha(\ell^\pi)=r, \alpha(m^\pi) = s \in [0,i]$, and $\alpha(n^\omega) =t \in [0,j]$.
We say $G$ satisfies $\phi$ under $\alpha$, denoted by $(G,\alpha) \models \phi$, if one of the following holds:
\begin{enumerate}[$\bullet$]
\item $\phi = E_a(\ell^\pi,m^\pi)$, $s=r+1$, and $a = a_s$.

\item $\phi = \ell^\pi < m^\pi$ and $r < s$.

\item $\phi = \ell^\pi \eqid n^\omega$ and $\idm(v_r)=\idm(w_t)$.

\item $\phi = \ell^\pi \eqdata n^\omega$ and $\nodedata(v_r)=\nodedata(w_t)$.

\item $\phi = \neg \psi$ and  it is not the case that $(G, \alpha) \models \psi$. 

\item $\phi = \phi' \vee \phi''$ and we have $(G,\alpha) \models \phi'$ or $(G,\alpha) \models \phi''$.

\item $\phi = \exists \ell^\pi. \psi$ and there is a node position $r$ in $\alpha(\pi)$ such that \linebreak
$(G, \alpha [\ell^\pi\mapsto r]) \models \psi$.
	
\item $\phi = \exists \pi. \psi$ and there is a path $\rho$ in $G$ such that
$(G,\alpha[\pi \mapsto \rho]) \models \psi$.
\end{enumerate}

\subsection{Querying graphs based on regular languages}\label{sec:regular_based_languages}

The language discussed in the preceding section has limitations in expressing certain aspects of path label patterns. Consider the query:
\begin{align*}
 Q_{\even}:=&\textit{``Is there a node labelled $a$ connected by a path of even length}\\ 
 & \textit{to a node labelled $b$ in graph $G$?''.}
\end{align*}

In {\wl} the query cannot be expressed \cite[Proposition 7]{Hellings:13:ICDT}. However, the set of paths satisfying this property can be described using a {\em regular path query} (\rp). An \rp comprises three elements: a source node $v_s$, a target node $v_t$, and a regular expression $e$ (or an NFA $A$). A pair of nodes $(v_s, v_t)$ satisfies such a query if there exists a path from $v_s$ to $v_t$ with the label of $\rho$ is in the language defined by $e$, that is $\lb{\rho}\in L(e)$.

Various extensions of regular path query languages have been introduced in the literature, initially for labelled graphs and later for graphs with data. For instance, conjunctive regular path queries, as in \cite{Lin:12:DS}, allow conjunction as well as automata operating over vectors of paths. In contrast, \cite{LV12} and \cite{BFL15} permit automata capable of detecting patterns involving data values.

We introduce a language that subsumes all languages in this paradigm, {\em first order extended regular data path query language} (\foerd).\footnote{In the literature, \foerd for graphs without data is known as ECRPQ$^\neg$ \cite{Lin:12:DS}.} 
We first describe the formalism informally. 
Regular languages are first expanded to languages defined by register automata.

\begin{definition}[Register automaton] \label{def:registerautomaton} %michael: do not comment out definition headers
A {\em register automaton} (RA) is an extension of finite automata, utilizing registers to store data values. A register automaton consists of a finite set of states $Q$, an initial state $q_0 \in Q$, and a set of final states $F \subseteq Q$, along with a finite set of registers.
A transition relation $\delta$ takes as input a state and a Boolean combination of comparisons between registers and constants, along with a special symbol for the input
data values. The (nondeterministic) output is a new state and optionally a set of update actions to each register, where an update sets a register
to another register's value, a constant, or the input data.
\end{definition}
Informally, when  processing input symbols from an infinite alphabet, register automata compare symbols with register values, updating both registers and automaton state. The semantics are straightforward: see \cite{Kaminski:94:TCS}.

 %A greg$A$ can be thought of as a fusion of an RA and an NFA. It has a finite set of registers, for keeping data values, %\michael{do we?} \dd{revised}
% along with two disjoint finite sets of states, \emph{word states} and {\em data states}. While in a word state, $A$ reads a symbol from $\Sigma$ and transitions to a data state. In a data state, it reads an element from $\mathcal{D}$, updates its registers, and reverts to a word state. It also has a distinguished initial and final word state, and  transitions that are either of the form $p \xrightarrow{a} q$, where $p$ is a word state, $q$ is a data state, and $a \in \Sigma$, or  $r \xrightarrow{\precondt, \updatet} s$, where $r$ is a data state and $s$ is a word state, while $\precondt$, the precondition, is a conjunction of equalities and inequalities among input data values, registers and constants. The update function $\updatet$ is a conjunction of assignments of register values, input data values, or constants to registers.
 
%The semantics of register automata is straightforward \cite{Kaminski:94:TCS}.
 
We now discuss \emph{register data path automata} ( {\ra}s ). These are  modifications of register automata to accept data paths. First, the state space is partitioned into
 \emph{word states} and {\em data states}. While in a word state, an \ra $A$ reads a symbol from $\Sigma$ and transitions to a data state: there is no update and comparison with registers.  In a data state, it reads an element from $\mathcal{D}$, updates its registers, and reverts to a word state. 
The initial and final states are a \emph{data state} and a \emph{word state}, respectively. Transitions that are either of the form $p \xrightarrow{a} q$, where $p$ is a word state, $q$ is a data state, and $a \in \Sigma$, or  $r \xrightarrow{\precondt, \updatet} s$, where $r$ is a data state and $s$ is a word state, while $\precondt$, the precondition, is a conjunction of equalities and inequalities among input data values, registers and constants. The update function $\updatet$ is a conjunction of assignments of register values, input data values, or constants to registers.we allow the input alphabet to be a product space, and thus in $\precondt$ and $\updatet$
we can refer to the $i^{th}$ input value. We can also fix infinite domains for each input value and register, and require $\precondt$ and $\updatet$ to be well-typed. Checking well-typedness -- e.g. we only compare registers with the same domain -- is straightforward. Thus in particular we will have certain registers that will only store ids, and we will assume that their domain is disjoint from the domains of other data values.

We now give the semantics more precisely. 
A data path automaton is simply an \ra where the input words have components corresponding to $\mathcal{D}_{\idm}$ and $\mathcal{D}_{\propm}$. Both the id and data domains contain a constant, denoted in each case by $\sharp$ (``unset''), which can be assigned to {\em id registers} $X_{\idm}$ and {\em $\datam$ registers} $X_{\datam}$, respectively. The semantics of a data path automaton will assume %enforce %\michael{Enforce, or assume in its semantics?} 
that the initial value of each register is $\sharp$.
Assuming $\mathcal{D}_{\propm}$ is a unary relation for simplicity, the data transitions will have a special shape, which we denote as $r \xrightarrow{E,I,U} s$, where:
\begin{itemize}
    \item $r$ is a data state,
    \item $s$ is a word state,
    \item $E = (E_{\idm}, E_{\datam})$, 
    \item $I = (I_{\idm}, I_{\datam})$, and
    \item $U = (U_{\idm}, U_{\datam})$.
\end{itemize}
Here:
\begin{itemize}
    \item $E_{\idm}$ (``$\idm$ equality check registers''), $I_{\idm}$ (``$\idm$ inequality check registers''), and $U_{\idm}$ (``$\idm$ update registers'') are subsets of $X_{\idm}$,
    \item $E_{\datam}$, $I_{\datam}$, and $U_{\datam}$ are subsets of $X_{\datam}$.
\end{itemize}

We now give the semantics.
 Assuming $\ell$ registers and a single variable domain $\mathcal{U}$, a {\em configuration} of an \ra $A$ is a pair $(q,v)$ where $q$ is a state, and $v \in (\mathcal{U})^\ell$ is an $\ell$-tuple representing register content. 
 A precondition corresponding to $E$ and $I$ above states that, when reading a symbol $(d_{\idm}, d_{\datam})$, we have  \begin{itemize}
\item $v[i]=d_{\idm}$ for each $i$ in $E_{\idm}$,
\item $v[i]=d_{\datam}$ for each $i$ in $E_{\datam}$,
\item $v[i]\neq d_{\idm}$ for each $i$ in $I_{\idm}$
\item $v[i]\neq d_{\datam}$ for each $i$ in $I_{\datam}$,
\end{itemize}
And the update corresponding to $U$ above is:
\begin{itemize}
\item $w[i]=d_{\idm}$ for each $i$ in $U_{\idm}$,
\item $w[i]=d_{\datam}$ for each $i$ in $U_{\datam}$,
\item $w[i]=v[i]$ for each $i \not\in U_{\idm}\cup U_{\datam}$.
\end{itemize}
Given two configurations $(p,v), (q,w)$ and a transition $\delta$, we write $(p,v) \linebreak \xrightarrow{(d_{\idm}, d_{\datam})}_\delta (q,w)$ if the precondition and update function are respected. $(q,v)$ is designated as the {\em initial configuration} if $q$ is the initial state and $v = (\sharp, \dots, \sharp)$. Similarly, $(q,v)$ is a {\em final configuration} if $q$ is a final state.
A sequence of configurations $s=c_0 \dots c_n$ is a {\em computation} on a data path $u=a_1 \dots a_n$ if there exists a sequence of transitions $\delta_1 \dots \delta_n$ such that $c_{i-1} \xrightarrow{a_i}_{\delta_i} c_i$ for each $i \in [n]$. If $c_0$ is the initial configuration and $c_n$ is a final configuration, then $s$ is an {\em accepting computation}, and $u$ is accepted by $A$. The language of regular data paths recognized by $A$, denoted as $L(A)$, comprises data paths $u$ accepted by $A$. 
Along the computation $s$, we say register $x_i$ is {\em written} at the $j^{th}$ step where configuration $c_j=(q_j,v_j)$ for $1 \leq j \leq n$ if $v_{j-1}[i]=\sharp$ and $v_j[i]\neq \sharp$.
Note that we can use states to keep track of which registers were written at the 
%(constantly many) 
previous states,
and whether the current content is $\sharp$.

\begin{figure}[t]
\centering
\begin{tikzpicture}[node distance=100pt]
\tikzstyle{state}=[draw,shape=circle,minimum size=18pt]

\node[initial, state] (q0) at (0,0) {$q_0$}; 
\node[state] (q1) [right of=q0] {$q_1$};
\node[state] (q2) [right of=q1] {$q_2$}; 
\node[accepting, state] (q3) [right of=q2] {$q_3$};

\path (q0)  edge [->] node[above] {$\emptyset,\emptyset,\{x\}$} (q1)
(q1)  edge [->, bend left] node[above] {$a$} (q2)
(q2)  edge [->, bend left] node[above] {$\emptyset,\emptyset,\emptyset$} (q1)
(q2)  edge [->] node[above] {$\{x\},\emptyset,\emptyset$} (q3)
(q3)  edge [->,loop above] node[above] {$a$} (q3)
(q3)  edge [->,loop below] node[below] {$\emptyset,\emptyset,\emptyset$} (q3)
;

\end{tikzpicture}
\caption{\ra $A$ over $\Sigma=\{a\}$ with one register $x \in X_\datam$.}
\label{fig:ra_ex}
\end{figure}

See Figure~\ref{fig:ra_ex}. The \ra $A$ has one register $x \in X_\datam$ and four states $q_0, \dots, q_3$, where $q_0$ is the initial state, and $q_3$ is a final state. While reading a data path $d_0 ~ a ~ d_1 \dots a ~ d_n$, it first stores the data value $d_0$ (excluding its id)  in the register $x$ and transitions from state $q_0$ to $q_1$. At state $q_2$, for each $i \in [1,n]$, when reading the data value $d_i$ (excluding its id), it checks whether the value is identical to the content of the register. If the values match, it transitions to the final state $q_3$ and accepts the data path.

An $n$-ary \ra $A$ is an \ra over an $n$-ary product alphabet.
Let $u_1,\dots, u_n$ be $n$ data paths, where $u_i=(d^{\idm,(i)}_0,d^{\datam,(i)}_0) \pathsep a^{(i)}_1\dots (d^{\idm,(i)}_{j_i},d^{\datam,(i)}_{j_i})$ for $1\leq i \leq n$ and some $j_1,\dots,j_n \in \mathbb{N}$. The {\em convolution} of $u_1,\dots, u_n$ is an $n$-ary data path $d_0a_1\dots d_m$, where:
\begin{itemize}
\item $m = \max\{j_1,\dots,j_n\}$,
\item $a_j=(b_1,\dots,b_n)$ where $b_i=a^{(i)}_j$ if $j \leq j_i$; otherwise, $b_i=\sharp$ for $1 \leq i \leq n$ and $1 \leq j \leq m$,
\item $d_j=(e^{\idm}_1,\dots,e^{\idm}_n,e^{\datam}_1,\dots,e^{\datam}_n)$ where $e^{\idm}_i=d^{\idm,(i)}_j$ and $e^{\datam}_i=d^{\datam,(i)}_j$ if $j \leq j_i$; otherwise $e^{\idm}_i=e^{\datam}_i=\sharp$ for $1 \leq i \leq n$ and $1 \leq j \leq m$.
\end{itemize}
The convolution of $u_1, \dots, u_n$ converts the $n$-tuple of data paths into a single data path of tuples. This resulting data path is as long as the longest data path among $u_1, \dots, u_n$ and is padded with $\sharp$'s for shorter paths. 
We say $(u_1,\dots, u_n)$ is accepted by $A$ if their convolution is accepted by $A$.

With this foundation, we proceed to define \edn.

\begin{definition}[\edn] \label{def:fo:erdpq}
The formulas of the logic are defined inductively:
\begin{itemize}
\item atom $::= \pi = \omega \mid x=y \mid (x,\pi,y) \mid (\pi_1,\dots,\pi_n) \in A$,
\item $\phi ::= \textit{atom} \mid \neg\phi \mid \phi \wedge \phi \mid \exists x. \phi \mid \exists \pi. \phi$,
\end{itemize}
Here, $\pi$ and $\omega$ are path variables, $x$ and $y$ are node variables, and $A$ is an $n$-ary {\em \ra} where $n\in \mathbb{Z}_{>0}$. For ease of understanding, we also write $(\pi_1, \dots, \pi_n) \not\in A$ to mean $\neg((\pi_1, \dots, \pi_n) \in A)$.
%, and $(\phi) \wedge (\phi')$ to mean $\phi \wedge \phi'$, etc. 
%\michael{I don't se why we write $(\phi) \wedge (\phi')$ and how this eases understanding of anything. And I also don't understand ``etc.''}
%\dd{removed}
\end{definition}

A query in \edn is referred to as a {\em first order regular data path query} (\ford) \footnote{In the literature, \ford for graphs without data is known as CRPQ$^\neg$ \cite{Lin:12:DS}.} if all {\ra}s in the query are unary. An \cdn query is called a {\em regular data path query} (\rd) if it is of the form: 
\[\exists\pi.(x,\pi,y) \wedge \pi \in A.\]
The semantics of \edn is standard. 

Let $\phi$ be an \edn formula, $G$ a data graph, and $\alpha$ be a mapping that assigns paths in $G$ to path variables and nodes to node variables. We say that $G$ satisfies $\phi$ under $\alpha$, denoted by $(G,\alpha) \models \phi$, if one of the following holds:
\begin{itemize}
\item $\phi: \pi = \omega$ and $\alpha(\pi) = \alpha(\omega)$.
\item $\phi: x=y$ and $\alpha(x) = \alpha(y)$.
\item $\phi: (x,\pi,y)$ and $\alpha(\pi)$ is a path from $\alpha(x)$ to $\alpha(y)$ in $G$.
\item $\phi: (\pi_1,\dots,\pi_n) \in A$ and $(\tdp{\alpha(\pi_1)}, \dots, \tdp{\alpha(\pi_n)})$ is in the $n$-ary relation of data paths defined by $A$.

\item $\phi: \neg \psi$ and $G$ does not satisfy $\psi$ under $\alpha$.

\item $\phi: \phi' \wedge \phi''$ and we have $(G,\alpha) \models \phi'$ and $(G,\alpha) \models \phi''$.

\item $\phi: \exists x. \psi$ and there is a node $r$ in $G$ such that $(G, \alpha [x\mapsto r]) \models \psi$.
	
\item $\phi: \exists \pi. \psi$ and there is a path $\rho$ in $G$ such that $(G,\alpha[\pi \mapsto \rho]) \models \psi$.
\end{itemize}

Obviously, $Q_{\even}$ is expressible in \edn. It corresponds to formula $\phi_{\even}(x,y):= \exists \pi. (x,\pi,y) \wedge \pi \in A_{\even}$, where $A_{\even}$ recognizes the set of data paths of even length.
Moreover, since ids are included as a part of data values, Hamiltonicity is expressible in \edn, as we show in the following example:

\begin{example}\label{ex:edn_Hamiltonian}
A path $\rho = v_0 \pathsep a_1 \dots v_n$ is simple if there are no $0 \leq i < j \leq n$ such that %michael: avoid s.t.
$v_i = v_j$. Let $A_{\repeatm}$ be
the unary \ra with one register which non-deterministically  records a node id and accepts the input if the recorded id repeats.  Consequently, a path $\rho$ in a data graph satisfies $(\pi \not\in A_{\repeatm})$ if and only 
if $\rho$ is simple. A path $\rho$ is said to visit all nodes if, for each node $v$, there exists an index $i$ 
such that $v_i=v$. Now, let $A_{\visit}$ be a binary \ra with one register, storing the first node id of the first data path and accepting if the stored id appears in the second path. 
Then, the \edn formula $\phi_H(\pi) := (\pi \not\in A_{\repeatm}) \wedge (\forall \omega. (\omega,\pi) \in A_{\visit})$ asserts that $\pi$ represents a Hamiltonian path.
\end{example}

\myparagraph{Graph pattern calculus (\gpc)}
%\michael{Paragraph header is not visible. Probably we should go back to using subsections if the substructure is large, like in this 
%case. For smaller things, we can use paragraph headers, but we need to redefine them (or use a custom version) so they are clearly visible}
%\subsubsection*{Graph pattern calculus (\gpc):}
\gpc is another query language based on regular expressions, designed for working with graphs containing data \cite{GPC:22:Francis}. 
As discussed in the introduction, we make some simplifications to the language proposed in the standard, related to inverses and trails. We return to these simplifications in Section \ref{sec:conc}.
%It should be noticed that, in this paper, we do not consider inverses in  query languages based on regular lajguages, including \gpc and extensions of \rd, for two main reasons:  
%\begin{inparaenum}[(i)]  
%\item Whether or not inverses are included in  query languages based on regular languages does not affect the relationships among the query languages in terms of expressiveness, as shown in Figure~\ref{fig:existing} and Figure~\ref{fig:summary}.  
%\michael{Not sure this is the place for this discussion}
%\item Introducing inverses into \rp results in the {\em two-way regular path query language} (\trp), which extends \rp by enabling navigation along graph edges in both directions \cite{Vardi:00:PODS}. It is known that every \trp query can be represented by an equivalent \rp-like query where the underlying NFA is replaced by a two-way NFA \cite{Calvanese:00:KR}. While both NFAs and two-way NFAs define regular languages and hold some desirable properties, such as decidable emptiness \cite{Hopcroft:79:book}, the emptiness problem for two-way register automata is undecidable, even for deterministic automata \cite{Neven:04:ACM}. Consequently, determining whether there exists a data graph that satisfies a query in \rd or its extensions becomes undecidable when inverses are allowed.  
%\end{inparaenum}
%\michael{I have lost the thread here. We were defining a language?}

The most foundational \gpc query  takes the form $\restrictor\ \pattern$, where $\restrictor$  is referred to as a {\em restrictor}, and $\pattern$ is called a {\em pattern}.
Informally, a pattern can describe the label pattern of a path, which can be represented equivalently as a regular expression. It can also perform equality tests on data values of two nodes. A restrictor can constrain paths to have non-repeating vertices, and can also enforce the paths to be the shortest among the valid paths.
The  forms of a restrictor $\restrictor$ are as follows:
\begin{itemize}
\item {\bf Restrictor:} $\restrictor::=  \simple  ~ | ~ \shortest  ~| ~ \shortestsimple$
%\footnote{In this article, we focus on data graphs instead of property graphs, so we do not utilize edge variables in patterns. Moreover, restrictors like trails or shortest trails, as in \cite{GPC:22:Francis}, are not applicable here.}
\end{itemize}

Atomic \gpc patterns are either:
\begin{itemize}
\item {\bf Node patterns:}
$()$ and $(x)$; or
\item {\bf Edge patterns:}
$\rightarrow$ and $\xrightarrow{a}$.
\end{itemize}
Here, $x$ represents a node variable, and $a \in \Sigma$. 

We now give inductive rules for forming patterns, assuming two patterns $\pattern$ and $\pattern'$. We close under:
\begin{itemize}
\item {\bf Union:} $\pattern + \pattern'$, 
\item {\bf Concatenation:} $\pattern\pattern'$, 
\item {\bf Repetition:} $\pattern^{n..m}$, where $n$ and $m$ can take values in $\mathbb{N}\cup\{\infty\}$.
\item {\bf Conditioning:} $\pattern_{\langle \theta \rangle}$, where $\theta$ is a condition.
\end{itemize}

Conditions are defined as:
\begin{itemize}
\item {\bf Condition:}
$\theta::= x \eqdata y \mid \neg\theta \mid \theta \wedge \theta$
\end{itemize}
where $x$ and $y$ are node variables.

%Finally we define \gpc queries via the recursive rule:
%\begin{itemize}
%\item {\bf Query:} $Q ::= \restrictor\ \pattern \mid Q,Q$
%\footnote{In \cite{GPC:22:Francis}, \gpc queries also include queries of the form $z = \restrictor\ \pattern$, where $z$ is a path variable. However, each path variable can appear at most once in a query, which intuitively implies that all path variables are independent. In this article, we focus on queries that return Boolean values. Accordingly, the presence or absence of path variables does not affect the expressiveness of the language.}
%end{itemize}
%where $\pattern$ is a pattern and $\restrictor$ is a restrictor.
%Given two queries $Q_1$ and $Q_2$, the query $Q_1,Q_2$ is referred to as the \emph{join} of $Q_1$ and $Q_2$.

In \cite{GPC:22:Francis}, \gpc is accompanied by a type system that imposes certain restrictions. Here, we consider the notion of {\em free variables}, as introduced in \cite{Gheerbrant:24:arxiv}. 
Given a pattern $\pattern$, a node variable $x$ in $\pattern$ is {\em bounded} if it appears within a repetition sub-pattern $\primepattern$; otherwise, $x$ is {\em free}. 
An important restriction  on the syntax is: 

\medskip

\emph{If a variable $x$ is used in the repetition sub-pattern $\primepattern$ of $\pattern$, it cannot appear anywhere else in $\pattern$ except within $\primepattern$}

We use $\fv{\pattern}$ to denote the set of free variables in the pattern $\pattern$.
\footnote{The definitions of \gpc provided in \cite{GPC:22:Francis} and \cite{Gheerbrant:24:arxiv} are not identical. For instance, the pattern $((x) + ()) \rightarrow (y)$ is valid in \cite{GPC:22:Francis} but not in \cite{Gheerbrant:24:arxiv}, because the node variable $x$ appears in the node sub-pattern $(x)$ but not in the sub-pattern $()$. In general, in \cite{Gheerbrant:24:arxiv} the use of variables is more complex.
%are additional restrictions: the union pattern $\pattern + \pattern'$ is valid only if $\fv{\pattern} = \fv{\pattern'}$.
In this paper, we adopt the simpler definitions from \cite{GPC:22:Francis}.}
%\item Since we consider only node variables, the sole restriction imposed on \gpc is that a variable in a pattern must be either free or bounded.
%\end{itemize}

Let $G = (V, E, \idm, \nodedata)$ be a data graph, and let $\pattern$ and $\pattern'$ be two \gpc patterns. For any two mappings $\mu: \fv{\pattern} \to V \cup \{\nullnode\}$ and $\mu': \fv{\pattern'} \to V \cup \{\nullnode\}$, the union $\mu \cup \mu': \fv{\pattern} \cup \fv{\pattern'} \to V \cup \{\nullnode\}$ is defined if, for all $x \in \fv{\pattern} \cap \fv{\pattern'}$, it holds that $\mu(x) = \mu'(x)$ whenever both $\mu(x) \neq \nullnode$ and $\mu'(x) \neq \nullnode$, where $\nullnode$ refers to a null value. 
If $\mu \cup \mu'$ is defined, then it maps each $x$ to:
\begin{itemize}
    \item $\mu(x)$ if $\mu(x) \neq \nullnode$ or $x \not\in \fv{\pattern'}$;
    \item $\mu'(x)$ if $\mu'(x) \neq \nullnode$ or $x \not\in \fv{\pattern}$;
    \item $\nullnode$ if $\mu(x) = \mu'(x) = \nullnode$.
\end{itemize}
Given a condition $\theta$ and a mapping $\mu: \fv{\pattern} \to V$, $\mu \models \theta$ is defined recursively as follows:
\begin{itemize}
    \item $\mu \models x \eqdata y$ if $\nodedata(\mu(x)) = \nodedata(\mu(y))$.
    \item $\mu \models \neg \theta$ if $\mu \not\models \theta$.
    \item $\mu \models \theta \wedge \theta'$ if $\mu \models \theta$ and $\mu \models \theta'$.
\end{itemize}

Now, we are ready to define the semantics of \gpc queries. Let $\pattern$ be a \gpc pattern, and let $G = (V, E, \idm, \nodedata)$ be a data graph over $\Sigma$. The semantics of $\pattern$ in $G$, denoted $\semantics{\pattern}{G}$, is a set of pairs $(\rho, \mu)$, where $\rho$ is a path and $\mu$ is a mapping from $\fv{\pattern}$ to $V$. The semantics is defined recursively as follows:
\begin{itemize}
    \item $\semantics{(x)}{G} = \{(v, \{x \mapsto v\}) \mid v \in V\}$. Recall that $v$ is a path of length zero in $G$ if $v \in V$.
    \item $\semantics{()}{G} = \{(v, \emptyset) \mid v \in V\}$.
    \item $\semantics{\xrightarrow{a}}{G} = \{(v \pathsep a \pathsep v', \emptyset) \mid (v, a, v') \in E\}$ for $a \in \Sigma$. Recall that $v \pathsep a \pathsep v'$ is a path of length one in $G$ if $(v, a, v') \in E$.
    \item $\semantics{\rightarrow}{G} = \{(v \pathsep a \pathsep v', \emptyset) \mid (v, a, v') \in E\}$.
    \item $\semantics{\pattern + \pattern'}{G} = \{(\rho, \mu') \mid (\rho, \mu) \in \semantics{\pattern}{G}$ or $(\rho, \mu) \in \semantics{\pattern'}{G}\}$, where $\mu'$ is a mapping from $\fv{\pattern} \cup \fv{\pattern'}$ defined as follows:
\begin{itemize}
    \item $\mu'(x) = \mu(x)$ if $x$ is in the domain of $\mu$;
    \item $\mu'(x) = \nullnode$ otherwise.
\end{itemize}
For example, suppose $\fv{\pattern}=\{x,y\}$ and $\fv{\pattern'}=\{y,z\}$. Then, we have $\fv{\pattern + \pattern'} = \{x,y,z\}$.  
Additionally, if $(\rho,\mu) \in \semantics{\pattern'}{G}$ with $\mu(y)=v$ and $\mu(z)=w$, then we have $(\rho,\mu') \in \semantics{\pattern \pattern'}{G}$, where:
\begin{itemize}
    \item $\mu'(x) = \nullnode$,
    \item $\mu'(y) = v$,
    \item $\mu'(z) = w$.
\end{itemize}

%\michael{I still don't understand this. Can you replace the right hand side with $(\rho, \mu) \in \semantics{\pattern}{G}$ or $(\rho, \mu) \in \semantics{\pattern'}{G}$? And is $\mu'$ unique given $\mu$? }
%\dd{Yes, $(\rho, \mu) \in \semantics{\pattern}{G} \cup \semantics{\pattern'}{G}$ means that $(\rho, \mu) \in \semantics{\pattern}{G}$ or $(\rho, \mu) \in \semantics{\pattern'}{G}$. The definition has been modified accordingly. Also, the mapping $\mu'$ is unique.}
%\michael{I still don't follow. "otherwise" usually means "if the other thing does not hold then...". And here what is the other thing? And does "map $x$ to nothing" mean that $x$ is not in the domain? I understood that the domain needs to be exactly the free variables of the expression, and every such variable needs to be mapped to a node}
%\dd{Maybe we should define what $f(x) = \text{nothing}$ means. Here, I follow the terminology of [11]. In [21], this issue does not arise because, in [21], the pattern $\pattern + \pattern'$ is valid only if $\fv{\pattern} = \fv{\pattern'}$. But does it make sense to define $f(x) = \text{nothing}$? }

%  $\semantics{\pattern + \pattern'}{G} = \{(\rho, \mu \cup \mu') \mid (\rho, \mu \cup \mu') \in \semantics{\pattern}{G} \mbox{ or} (\rho, \mu \cup \mu') \in \semantics{\pattern'}{G}\}$.
    \item $\semantics{\pattern\pattern'}{G} = \{(\rho \cdot \rho', \mu \cup \mu') \mid (\rho, \mu) \in \semantics{\pattern}{G}, (\rho', \mu') \in \semantics{\pattern'}{G}\text{, and the union } \mu\cup\mu' \text{ is defined}\}$.
    \item $\semantics{\pattern^{n..m}}{G} = \bigcup_{i=n}^m \semantics{\pattern}{G}^i$, where:
    \begin{itemize}
        \item $\semantics{\pattern}{G}^0 = \semantics{()}{G}$.
        \item $\semantics{\pattern}{G}^\ell = \{(\rho_1 \cdots \rho_\ell, \emptyset) \mid \exists \mu_1, \dots, \mu_\ell. \, (\rho_1, \mu_1), \dots, (\rho_\ell, \mu_\ell) \in \semantics{\pattern}{G}\}$ for $\ell \geq 1$.
    \end{itemize}
    \item $\semantics{\pattern_{\langle \theta \rangle}}{G} = \{(\rho, \mu) \mid (\rho,\mu) \in \semantics{\pattern}{G} \text{ and }\mu \models \theta\}$.
\end{itemize}

Given a \gpc pattern $\pattern$ and a restrictor $\restrictor$, the semantics $\semantics{\restrictor \ \pattern}{G}$ is defined recursively as follows:
\begin{itemize}
    \item $\{(\rho, \mu) \mid (\rho, \mu) \in \semantics{\pattern}{G}  \rho \text{ is simple}\}$ if $\restrictor = \simple$.
    \item $\{(\rho, \mu) \mid (\rho, \mu) \in \semantics{\pattern}{G} ~  \rho \text{ shortest among paths } \rho' 
    \mbox{ with } (\rho', \mu') \in \semantics{\pattern}{G} \\ \mbox{ and } \rho' \mbox{ has the same endpoints as } \rho \}$ if $\restrictor = \shortest$. %\michael{Reviewer comment (which seems fair to me): ``shortest among what?''}
   
    \item $\{(\rho, \mu) \mid (\rho, \mu) \in \semantics{\pattern}{G}, \rho \text{ is simple, and } |\rho| = \min \{ |\rho'| \mid (\rho', \mu') \in \semantics{\pattern}{G},\linebreak \rho' \text{ is simple and has the same endpoints as } \rho \}$ if $\restrictor =\shortestsimple$.
    
\end{itemize}

%\michael{Reviewer comment: ``Before you start defining the semantics of joins, explain that now the
%semantics stores tuples of paths, rather than single paths.'' (but I haven't really reread this part to see if the comment makes sense}
%For two \gpc queries $Q$ and $Q'$ with semantics $\semantics{Q}{G}$ and $\semantics{Q'}{G}$, respectively, we define:
\begin{itemize}
    \item $\semantics{Q, Q'}{G} = \{((\bar{\rho}, \bar{\rho}'), \mu \cup \mu') \mid (\bar{\rho}, \mu) \in \semantics{Q}{G}, (\bar{\rho}', \mu') \in \semantics{Q'}{G} \text{, and the union }\linebreak \mu\cup\mu' \text{ is defined}\}$.
\end{itemize}
We omit the subscript $G$ when it is clear from the context. We adapt \gpc queries to support graph queries in the obvious way: a non-empty data graph $G$ satisfies a \gpc query $Q$ if $\semantics{Q}{G}$ is non-empty.
For simplicity, we also say that a path $\rho$ in $G$ is in the semantics of $Q$, denoted as $\rho \in \semantics{Q}{G}$, if $(\rho, \mu) \in \semantics{Q}{G}$.
%\begin{definition}{[Satisfaction of a pattern]} 
We also say that a path $\rho$ satisfies $\pattern$, if there exists a mapping $\mu$ such that $(\rho,\mu) \in \semantics{\pattern}{G}$.
%\end{definition}

We end this subsection with two examples.  Consider the following \gpc pattern over $\Sigma=\{a,b\}$:
\[
\pattern = (x) \xrightarrow{a} () \left[ \left[\xrightarrow{b}\xrightarrow{b}\xrightarrow{b}\right]^{1..\infty} + \left[[(y) \rightarrow^{3..7} (z)]_{y\eqdata z}\right]^{1..1} \right] () \xrightarrow{b} (x).
\]
Since variables $y$ and $z$ are used in a repetition sub-pattern, we have $\fv{\pattern} = \{x\}$.
Pattern $\pattern$ refers to the set of paths $\rho = v_0  ~ a_1 ~ v_1 \dots a_n ~ v_n$, where:
\begin{itemize}
    \item $v_0 = v_n$ where $n>2$,
    \item $a_1 = a$ and $a_n = b$, and
    \item at least one of the following holds:
    \begin{enumerate}[(a)]
        \item $n-2$ is divisible by $3$ and $a_i = b$ for $2\leq i \leq n-1$;
        \item $3\leq n-2 \leq 7$ and $\dataof(v_1) = \dataof(v_{n-1})$.
    \end{enumerate}
\end{itemize}

See the data graph $G$ in Figure~\ref{fig:gpc_eg}.
Let $\mu$ be a mapping with $\mu(x)=n_3$.
By definition, we have $\semantics{\pattern}{G}$ as the union of the following two sets:
\begin{enumerate}[(a$'$)]
    \item $P_1 = \{(n_3 ~ a ~ (n_4 ~ b ~ n_5 ~ b ~ n_6 ~ b ~ n_7~  b)^i ~ n_3, \mu) \mid 4i-1 \text{ is divisible by 3 and } i\geq 1\}$.
    \item $P_2 = \{(n_3 ~ a ~ n_1 ~ (a ~ n_2 ~  a  ~n_1)^i ~ a ~ n_2 ~  b ~ n_3, \mu) \mid i=2,\dots, 6\}$.
\end{enumerate}

Moreover, we find that $\semantics{\restrictor\ \pattern}{G}$ is as follows:
\begin{itemize}
\item When $\restrictor = \simple, \shortestsimple$, $\semantics{\restrictor\ \pattern}{G}$ is empty, as there is no simple path in $G$ matching pattern $\pattern$.

\item When $\restrictor = \shortest$, $\semantics{\restrictor\ \pattern}{G} = \{(\rho,\mu) \mid \rho \text{ is a shortest path in } P_1\cup P_2\}$.
\end{itemize}

\begin{figure}[!h]
\centering
\includegraphics[width=0.6\textwidth]{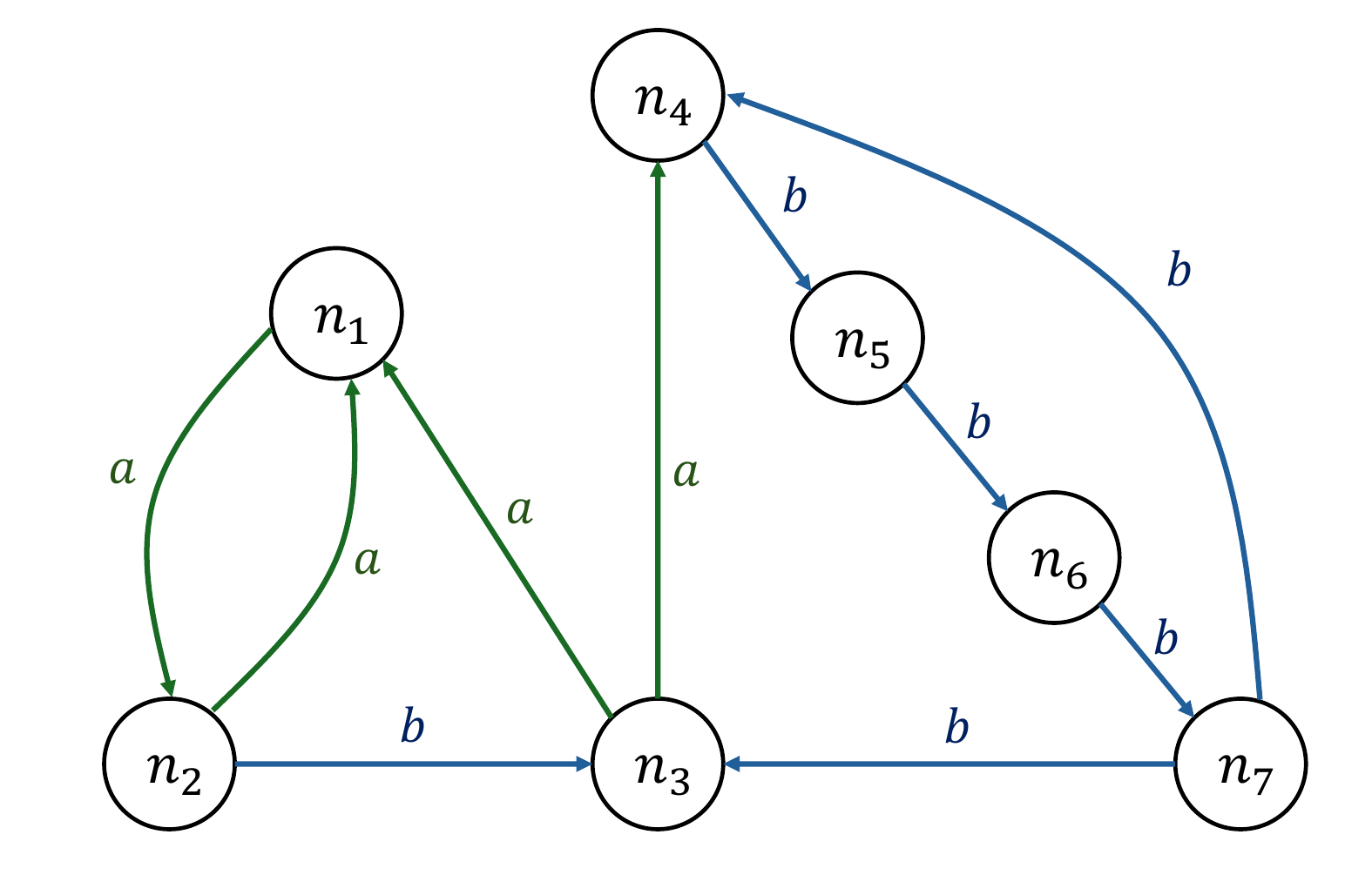}
\caption{Data graph $G$ over $\{a,b\}$, where $\dataof(n_1)=\dataof(n_2)=0$ and $\dataof(n_i)=i$ for $i=3,\dots,7$.}
\label{fig:gpc_eg}
\end{figure}

%\michael{Reviewer comment:  ``Oddly enough, you picked an example where all paths have the same
%endpoints, so the issue of grouping by endpoints does not arise. Please,
%adjust the example so that paths with different endpoints are present.'' (we could also add a second example)}
%\dd{Other than satisfying the reviewer's request, providing another example with different endpoints does not improve our paper. Our %proofs specifically design graphs and queries to ensure this scenario is avoided.}
%\michael{I think the issue of grouping with endpoints is not about concatenation. It is to illustrate the subtlety with shortest and %shortestsimple --- which we messed up.  }
%\dd{A new example is given.}

Finally, consider the following \gpc pattern over $\Sigma=\{a,b\}$:
\[\primepattern =\ \xrightarrow{b}\xrightarrow{b}.\]
Given the same data graph $G$ in Figure~\ref{fig:gpc_eg}, we have:
\[\semantics{\primepattern}{G} = \{n_4~ b ~ n_5~ b ~ n_6, n_5 ~ b ~ n_6 ~ b ~ n_6, n_6 ~ b ~ n_7 ~ b ~ n_4, n_6 ~ b ~ n_7 ~ b ~ n_3, n_7 ~ b ~ n_4 ~ b ~ n_5\}.\]
Since all paths in $\semantics{\primepattern}{G}$ are simple and have the same length, we have $\semantics{\restrictor\ \primepattern}{G}=\semantics{\primepattern}{G}$ for any restrictor $\restrictor$.

\subsection{Graph querying based on \fotc}
Another category of languages revolves around first-order logic with only node variables. Instead of relying on operations rooted in automata or regular languages, these languages utilize a transitive closure operator.

A data graph $G=(V,E, \idm, \nodedata)$ over $\Sigma$ can be viewed as a first-order structure $\langle V, \{E_a\}_{a \in \Sigma}, \eqdata \rangle$. First-order logic over this structure is denoted as \fod. In this structure, the set $V$ of nodes forms the universe, $E_a$ is interpreted as $\{(v,v') \in V^2 \mid (v,a,v')\in E\}$, and $\eqdata$ is interpreted as $\{(v,v') \mid \nodedata(v)=\nodedata(v')\}$.

We extend \fod by introducing a  transitive closure constructor.
If $\phi(\mathbf{x},\mathbf{y})$ is a formula, then $\phi^*(\mathbf{x}, \mathbf{y})$ becomes a formula, where $\mathbf{x}$ and $\mathbf{y}$ represent two vectors of node variables of the same arity. The interpretation is such that $G, \rho \models \phi^*(\mathbf{x}, \mathbf{y})$ for a data graph $G$ and a valuation $\rho$, if there exists a sequence of node vectors $\mathbf{x}_1=\rho(\mathbf{x}), \mathbf{x}_2, \ldots, \mathbf{x}_n=\rho(\mathbf{y})$ satisfying the condition that for each $1 \leq i < n$, $G, \rho \models \phi(\mathbf{x}_i, \mathbf{x}_{i+1})$.
\footnote{In some prior literature \cite{Libkin:04:text}, transitive closure operators are defined over formulas of the form $\phi(\mathbf{x}, \mathbf{y}, \mathbf{z})$, which involve three tuples of variables: $\mathbf{x}$, $\mathbf{y}$, and $\mathbf{z}$, where $\mathbf{z}$ represents free variables in $\phi^*(\mathbf{x}, \mathbf{y}, \mathbf{z})$. These two variants of the language \fotc are equivalent in expressiveness.
Specifically, $\phi^*(\mathbf{x}, \mathbf{y}, \mathbf{z})$ from the alternative variant can be translated into the equivalent formula $\exists \mathbf{z}'. {\phi'}^*(\mathbf{x}', \mathbf{y}')$ of the variant used in this paper, where $\mathbf{x}' = (\mathbf{x}, \mathbf{z})$, $\mathbf{y}' = (\mathbf{y}, \mathbf{z}')$, and $\phi'(\mathbf{x}', \mathbf{y}') = \phi(\mathbf{x}, \mathbf{y}, \mathbf{z}) \wedge \mathbf{z} = \mathbf{z}'$. }
%\michael{I do not understand these sentences. I think the claim is that it does not make any difference for full \fotc, although I do not follow the explanation. But there is also the question of whather it makes a difference for the fragments in the next paragraph}
%\dd{The sentences have been revised, but this does not affect the fragments in the next paragraph. This footnote has been added in response to a reviewer’s comment. It seems that he thinks these two variants have different expressive power.}
%}
We denote the extension of \fod with transitive closures as \fotcd. In contrast to other languages that rely on transitive closure of edge relations,\linebreak \fotcd permits the transitive closure of any definable relationship, not just the closure of the edge relation.

Within this paradigm, one prominent language family is GXPath, originating from  query languages for graph databases initially introduced in \cite{LMV16}, inspired by the XML query language XPath. Within this family,\linebreak \gpcrd is the most powerful variant:  \cite{LMV16} shows that it is equivalent to \fotcr, a restriction of \fotcd allowing only  formulas in which every subformula has at most three free variables, and the closure $\phi^*$ of $\phi(\mathbf{x},\mathbf{y})$ is allowed only if $\mathbf{x}$ and $\mathbf{y}$ have the same arity. 
%\michael{I don't understand "unary variables"}
%\dd{revised} michael: ok
Lastly, it is worth mentioning {\em regular queries} (RQ) \cite{Vardi:17:TCS}, a query language for graphs without data, which extends {\em unions of conjunctive 2-way regular path queries} (UC2RPQ) by adding transitive closure. RQ can be defined as the closure of atomic queries under selection, projection, disjunction, conjunction, and transitive closure. By adding $x \eqdata y$ to the set of atomic queries, RQ can be extended to handle graphs with data.

\section{Expressiveness of existing languages}\label{sec:prior}

The languages discussed in Section~\ref{sec:queries} are illustrated in Figure~\ref{fig:existing}. This section will be devoted to proving the following theorem, which justifies Figure~\ref{fig:existing}:

\begin{theorem} \label{thm:existingcontainments}
The subsumptions depicted in Figure~\ref{fig:existing} all hold.
\end{theorem}

It is straightforward, just chasing the definitions, to observe that \rd $\subseteq$ \cdn $\subseteq$ \edn. 
To complete the proof, we still have to show:
\begin{enumerate}[(1)]
\item \edn subsumes \gpc.
\item \edn subsumes \wl.
\item \fotcd subsumes \rd.
\end{enumerate}

We will show the first and third item here. The proof of the middle item will be deferred to Section \ref{sec:mwl}.

\subsection{\edn subsumption of \gpc}\label{sec:edn_subsumes_gpc}
We define a fragment of {\edn} with lower data complexity which still subsumes \gpc. The motivation is that there exists a substantial complexity gap between \gpc and \edn. The data complexity of query evaluation for \gpc is within \pspace \cite{GPC:22:Francis}, whereas for \edn, it is non-elementary \cite{Lin:12:DS}. Our objective is to bridge this gap by introducing the fragment \uedn of \edn, in which only universal quantifiers are allowed. The expressiveness of this fragment relative to others is depicted in Figure~\ref{fig:uedn}.

\begin{figure}[t]
\centering
\scalebox{0.7}{
\begin{tikzpicture}
\centering
[node distance=160pt]

\tikzset{invisible/.style={minimum width=0mm,inner sep=0mm,outer sep=0mm}}

\node (v1) at (0,2) {\edn};
\node (v2) at (0,1) {\uedn};
\node (v3) at (0,0) {\gpc};

%p
\path
(v1)
    edge [-] node[above, rotate=-90] {} (v2)
(v2)
    edge [-] node[above, rotate=-90] {} (v3)
;

\end{tikzpicture}
}
\caption{Fragment \uedn between \edn and \gpc. (Higher is more expressive.)}
\label{fig:uedn}
\end{figure}

Formalizing the syntax, a \uedn formula takes the form:
\[\forall \bar{\omega}. \psi(\bar{\pi},\bar{\omega}),
\]
where $\psi$ is a quantifier-free \edn formula. We will demonstrate that this fragment contains \gpc.

\begin{theorem}\label{thm:ednsubsumesgpc} Every \gpc query is expressible in \uedn.
\end{theorem}

%Recall that \gpc queries are defined recursively: 
%\[Q ::=  \restrictor\ \pattern \mid Q,Q\]
%It follows that a \gpc query $Q$ is a join of $n$ \gpc queries $Q_1,\dots,Q_n$ for some $n \in \mathbb{Z}_{>0}$, where $Q_i$ is of the form $\restrictor\ \pattern$ where $\restrictor$ is a restrictor and $\pattern$ is a pattern.
%We will prove Theorem \ref{thm:ednsubsumesgpc} by induction on $n$.
%We first consider the base case, $n=1$:

%\begin{proposition}\label{prop:gpctoaut}
%Let $Q$ be a \gpc query of the form $\restrictor\ \pattern$. There exists a 1-ary \uedn formula $\phi_Q(\pi)$ equivalent to $Q$. That is, for every data graph $G$ and path $\rho$ in $G$, $G \models \phi_Q(\rho)$ if and only if $\rho$ is in the semantics of $Q$.
%\end{proposition}
\begin{proof}
Let $Q$ be a \gpc query of the form $\restrictor\ \pattern$. We will show that there exists a 1-ary \uedn formula $\phi_Q(\pi)$ equivalent to $Q$. That is, for every data graph $G$ and path $\rho$ in $G$, $G \models \phi_Q(\rho)$ if and only if $\rho$ is in the semantics of $Q$.

This statement follows directly from the following claim:

%Recall that the semantics of the pattern $\pattern$ in a data graph $G$ is a set of pairs of the form $(\rho,\mu)$, where $\rho$ is a path in $G$, and $\mu$ maps each variable in $\fv{\pattern}$ to a node in $G$ or the null value $\nullnode$. 

\subparagraph*{Claim (S)} There is a 1-ary \ra $A_{\pattern}$ recognizing the data paths $\tdp{\rho}$ corresponding to the paths $\rho$ satisfying the pattern $\pattern$. That is, for every path $\rho$, $\tdp{\rho}$ is accepted by $A_{\pattern}$ if and only if there exists a mapping $\mu$ such that $(\rho,\mu) \in \semantics{\pattern}{G}$.

\medskip

Before justifying why Claim (S) holds, we first explain how it leads to the proposition.  
Suppose (S) is true. Then, we can construct the \uedn formula $\phi_Q$ as required in the proposition. There are three cases to consider: $\restrictor = \shortest, \simple,\shortestsimple$.   
Assuming the existence of the \ra $A_{\pattern}$ in (S), consider the case where $\restrictor = \shortest$. The required \uedn formula $\phi_Q(\pi)$ in this case is:
\[
\phi_{\shortest}(\pi) := \pi \in A_{\pattern} \wedge \forall \omega. (\omega \in A_{\pattern} \wedge A_{same}(\pi, \omega) \rightarrow (\pi, \omega) \in A_{\leq},
\]
%\[
%\phi_{\shortest}(\pi) :=  \forall \omega. (\pi \in A_{\pattern} \wedge \omega \in A_{\pattern} \wedge A_{same}(\pi, \omega)) \rightarrow (\pi, \omega) \in A_{\leq},
%\]
where $(\pi, \omega) \in A_{\leq}$ indicates that the length of $\pi$ is not greater than the length of $\omega$, while $A_{same}$ holds of a pair if they have the same endpoints. These two relations are clearly \ra definable: in the case of $A_{same}$, we use the fact that our attributes include an identity data value.
Further, the conjunct $\pi \in A_{\pattern}$ can be moved inside the universal quantification.

Consider $\restrictor = \simple$. A path $\rho$ is simple if for every two prefixes $\tau, \nu$ of $\rho$ with $|\tau| \neq |\nu|$, their last nodes have different ids. Using closure of $\ra$ under union, we can create a ternary \ra $A_{\simple}$ with one register to recognize the set of triples $(\tdp{\rho}, \tdp{\tau}, \tdp{\nu})$, such that:

\medskip

If
$\tau$ and $\nu$ are prefixes of $\rho$, and $|\tau| \neq |\nu|$, then
 the id of the last node in $\tau$ is different from the id of the last node in $\nu$.

\medskip

Consequently, a path is simple if and only if it satisfies the following \uedn formula:
\[
\phi_{\simple}(\pi) := \forall \omega, \theta. (\pi, \omega, \theta) \in A_{\simple}.
\]
Accordingly, the \uedn $\phi_Q(\pi)$ is:
\[\phi_{\simple}(\pi) \wedge \pi\in A_{\pattern}\]
assuming the \ra $A_{\pattern}$ from (S).

 Finally, for the case $\restrictor = \shortestsimple$, the \uedn formula $\phi_Q$ can be derived in a similar manner.
 Specifically, $\phi_Q(\pi)$ in this case is:
 \[\begin{aligned}
&\phi_{\shortestsimple}(\pi):= \forall \pi', \omega, \theta. \ \   \big( (\pi, \omega, \theta) \in A_{\simple} \wedge \pi \in A_{\pattern} \big) \ \wedge \\
& \big( (\pi' \in A_{\pattern} \wedge (\pi',\pi)\in A_< \wedge (\pi, \pi') \in A_{same}) \rightarrow \pi' \in A_{\nonsimple} \big),
\end{aligned}\]
where $(\pi', \pi) \in A_<$ indicates that $\pi'$ is strictly shorter than $\pi$, $(\pi, \pi') \in A_{same}$ indicates they have the same endpoints, while $\pi' \in A_{\nonsimple}$ expresses that $\pi'$ contains at least one repeated node id. Note that while simplicity is not \ra-definable, its complement $A_{\nonsimple}$ is recognizable by a 1-register \ra that non-deterministically stores a node id and verifies its recurrence.
%\michael{We also have to say $\pi$ is simple?}
%\dd{Yes, we do. $\forall \omega,\theta. (\pi,\omega,\theta)\in A_{\simple}$ implies that $\pi$ is simple.}
%michael{It is strange that this is inside the universal quantification}
Thus, (S) implies the proposition.
% by combining the conditions for $\shortest$ and $\simple$. Thus, (S) implies the proposition.

\bigskip

Each pattern $\pattern$ is recursively constructed from atomic patterns. If no node patterns of the form $(x)$ appear in $\pattern$, then $\pattern$ is simply a regular expression, and the construction process for $A_{\pattern}$, as stated in (S), is identical to the process of transforming a regular expression into a corresponding finite automaton. 
A key distinction arises when node variables $x$ are present in $\pattern$.

%\dd{The example has been removed.}

\bigskip

Now, we proceed to prove (S). For a given pattern $\pattern$, it is either an atomic pattern, a union or a concatenation of two sub-patterns, a repetition pattern, or a conditioning. We prove (S) by structural induction on patterns:

\paragraph*{Node patterns}
\begin{itemize}
\item If $\pattern = ()$, then $A_{\pattern}$ accepts all data paths of length zero. Specifically, $A_{\pattern}$ has two states $s_0$ and $s_1$, where $s_0$ is the initial state and $s_1$ is a final state. The automaton has no registers. On a data path $d_0 a_1 d_1 \dots a_n d_n$, it switches from state $s_0$ to $s_1$ when reading the data value $d_0$. Thus, the data path is accepted by $A_{\pattern}$ if and only if $n = 0$. Note that $A_{\pattern}$ is not required to be {\em complete}, meaning that it does not necessarily have a transition for every possible input from each state.

Obviously, in this case, for every path $\rho$, we have $(\rho,\emptyset) \in \semantics{\pattern}{G}$ if and only if $\tdp{\rho}$ is accepted by $A_{\pattern}$.

\item If $\pattern = (x)$, then $A_{\pattern}$ accepts all data paths $d$ of length zero, where $d = (d_{\idm}, d_{\datam}) \in \mathcal{D}$, and the automaton indicates that registers $r^{\idm}_x$ and $r^{\datam}_x$ are written with values $d_{\idm}$ and $d_{\datam}$, respectively. $A_{\pattern}$ is obtained by modifying the automaton from the previous case to store the data value $d_0$ in $r^{\idm}_x$ and $r^{\datam}_x$ during the transition from $s_0$ to $s_1$.

In this case, for every path $\rho$, if $(\rho,\mu) \in \semantics{\pattern}{G}$, then $\tdp{\rho} = d$ for some $d = (d_{\idm}, d_{\datam}) \in \mathcal{D}$. Furthermore, $(\rho,\mu) \in \semantics{\pattern}{G}$ with $\mu(x) = d_{\datam}$ if and only if there exists an accepting run $\gamma$ of $A_{\pattern}$ on the data path $\tdp{\rho}$ such that the value of register $r^{\datam}_x$ becomes $d_{\datam}$ along $\gamma$.
\end{itemize}

\paragraph*{Edge patterns}
\begin{itemize}
\item For $\pattern = \rightarrow$, $A_{\pattern}$ accepts all data paths of length one. Specifically, $A_{\pattern}$ has four states $s_0, s_1, s_2, s_3$, where $s_0$ is the initial state and $s_3$ is a final state. For a data path $d_0 a_1 d_1 \dots a_m d_m$, the transitions are as follows:
  \begin{enumerate}
  \item Transition from $s_0$ to $s_1$ on reading $d_0$.
  \item Transition from $s_1$ to $s_2$ on reading any symbol $a \in \Sigma$.
  \item Transition from $s_2$ to $s_3$ on reading $d_1$.
  \end{enumerate}
  Accordingly, $A_{\pattern}$ accepts the data path if and only if $m = 1$.

\item For $\pattern = \xrightarrow{a}$, $A_{\pattern}$ accepts all data paths $d_0 a_1 d_1$ of length one, where $a_1 = a$. The automaton for this case is obtained by modifying the transitions of $A_{\pattern}$ for $\pattern = \rightarrow$. Specifically, all transitions from $s_1$ to $s_2$ on symbols $b \neq a$ are removed, leaving only the transition on $a$.
\end{itemize}
In the case of edge patterns, for every path $\rho$, we have $(\rho,\emptyset) \in \semantics{\pattern}{G}$ if and only if $\tdp{\rho}$ is accepted by $A_{\pattern}$.

\paragraph*{Union}
Let's assume $\pattern = \pattern_1 + \pattern_2$, and we have the corresponding register automaton $A_{\pattern_1}$ and $A_{\pattern_2}$ for $\pattern_1$ and $\pattern_2$ respectively. The automaton $A_{\pattern}$ results from the disjunction of $A_{\pattern_1}$ and $A_{\pattern_2}$. Since register automata are closed under disjunction, the disjunction of $A_{\pattern_1}$ and $A_{\pattern_2}$ is well-defined. 

For every path $\rho$, if $\gamma$ is an accepting run of $A_{\pattern}$ on the data path $\tdp{\rho}$, we define $\mu$ as the mapping that assigns a free variable $x$ to $v$ if the registers corresponding to $x$ hold the value $v$, and to the null value $\nullnode$ if those registers remain unwritten along $\gamma$. 

Obviously, if $\gamma$ is an accepting run of $A_{\pattern}$ on $\tdp{\rho}$, then $\gamma$ is also an accepting run of $A_{\pattern_1}$ or $A_{\pattern_2}$ on $\tdp{\rho}$. Accordingly, if $(\rho,\mu_1) \in \semantics{\pattern_1}{G}$ or $(\rho,\mu_2) \in \semantics{\pattern_2}{G}$, where the domain of $\mu_i$ is $\fv{\pattern_i}$ and $\mu_i(x) = \mu(x)$ for all $x \in \fv{\pattern_i}$, for $i=1,2$, then it follows that $(\rho,\mu) \in \semantics{\pattern}{G}$ and $\mu=\mu_1 \cup \mu_2$.

\paragraph*{Concatenation}
Let's assume $\pattern = \pattern_1 \pattern_2$, and we have the corresponding register automaton $A_{\pattern_1}$ and $A_{\pattern_2}$ for $\pattern_1$ and $\pattern_2$, respectively. Without loss of generality, we assume that the final state of $A_{\pattern_1}$ is unique and has no outgoing transitions. The automaton $A_{\pattern}$ is obtained by concatenating $A_{\pattern_1}$ and $A_{\pattern_2}$. Specifically, this concatenation is achieved by introducing transitions defined as follows: suppose $r \xrightarrow{E,I,U} s$ is a transition of $A_{\pattern_1}$, where $s$ is the final state, and $r' \xrightarrow{E',I',U'} s'$ is a transition of $A_{\pattern_2}$, where $r'$ is the initial state. Then, $r \xrightarrow{E,I,U} s'$ is a transition of $A_{\pattern}$. 

Furthermore, we must modify the use of registers corresponding to variables $x$ that are present in both $\pattern_1$ and $\pattern_2$. Before explaining this modification, we note that no variable $x$ can appear in a repetition subpattern of $\pattern_1$ or $\pattern_2$, by the syntactic requirements of \gpc. Thus, if a variable $x$ occurs in both subpatterns $\pattern_1$ and $\pattern_2$, it must be free in both.
 
We will only modify the register behavior in $A_{\pattern_2}$, leaving the behavior in $A_{\pattern_1}$ unchanged. Informally, the modification will be to ensure the following property: 

\medskip

In the copy of $A_{\pattern_2}$, if we reach a state where
 the value of $x$ is already set to be $x_0$ and we are to make a transition where $x$ is to be written again, we only allow this if the data value in the input matches $x_0$. Otherwise the behavior is the same as in $A_{\pattern_2}$.

 \medskip
 
 One way to enforce this is to add a component of the state tracking whether each shared variable is written or not. It is straightforward with this additional state to ensure the property above.

If $\gamma=\gamma_1\cdot \gamma_2$ is an accepting run of $A_{\pattern}$ on $\tdp{\rho}$, then $\gamma_1$ and $\gamma_2$ are accepting runs of $A_{\pattern_1}$ and $A_{\pattern_2}$ on $\tdp{\rho_1}$ and $\tdp{\rho_2}$, respectively. Accordingly, if $(\rho,\mu_1) \in \semantics{\pattern_1}{G}$ and $(\rho,\mu_2) \in \semantics{\pattern_2}{G}$, where the domain of $\mu_i$ is $\fv{\pattern_i}$ and $\mu_i(x) = \mu(x)$ for all $x \in \fv{\pattern_i}$, for $i=1,2$, then it follows that $(\rho,\mu) \in \semantics{\pattern}{G}$ and $\mu=\mu_1 \cup \mu_2$.

On the other hand, if $(\rho_1,\mu_1) \in \semantics{\pattern_1}{G}$ and $(\rho_2,\mu_2) \in \semantics{\pattern_2}{G}$ with $\mu_1(x) = \mu_2(x)$ for all $x \in \fv{\pattern_1} \cap \fv{\pattern_2}$, and the last node of $\rho_1$ is identical to the first node of $\rho_2$, then there exist accepting runs $\gamma_1$ and $\gamma_2$ of $A_{\pattern_1}$ and $A_{\pattern_2}$ on $\tdp{\rho_1}$ and $\tdp{\rho_2}$, respectively, and $\tdp{\rho}=\tdp{\rho_1}\cdot\tdp{\rho_2}$ is a valid data path. Hence, $\gamma=\gamma_1\cdot \gamma_2$ is an accepting run of $A_{\pattern}$ on $\tdp{\rho}$.

%Note that given two arbitrary sets of data paths $P_1$ and $P_2$, the concatenation of $P_1$ and $P_2$ is not necessarily a set of valid data paths. For instance, if $\rho_1 = v_0 ~ a_1 \dots v_m \in P_1$ and $\rho_2 = v'_0  ~ a'_1 \dots v'_n \in P_2$, their concatenation $\rho_1 \cdot \rho_2$ is not a valid path when $v_m \neq v'_0$. This, however, is not an issue for our recursive construction of $A_\pattern$, because the input to an \ra is always a single valid data path.

\paragraph*{Repetition}
Given a pattern $\subpattern$, $\subpattern^{n \dots m}$ is also a pattern, where $n \in \mathbb{N}$ and $m \in \mathbb{N}\cup \{\infty\}$. This notation signifies that pattern $\subpattern$ is repeated between $n$ and $m$ times.
When both $n$ and $m$ belong to $\mathbb{N}$, then $\subpattern^{n \dots m}$ is equivalent to $\subpattern_1^{n \dots n} + \ldots + \subpattern_m^{m \dots m}$, where $\subpattern_i$ is a copy of $\subpattern$ with fresh node variables for $i\in [1,m]$.
Similarly, pattern $\subpattern^{n \dots \infty}$ is equivalent to $\subpattern_1^{n \dots n}  (() +\subpattern_2^{1 \dots \infty})$, where $\subpattern_i$ is a copy of $\subpattern$ with fresh node variables for $i=1,2$. Recall that $\pattern^{0..0}$ is equivalent to node pattern $()$ for all patterns $\pattern$. Therefore, we only need to consider the cases $\pattern=\subpattern^{n \dots n}$ and $\pattern=\subpattern^{1 \dots \infty}$.

Let $A_{\subpattern}$ be the \ra corresponding to $\subpattern$.
If $\pattern=\subpattern^{n \dots n}$, then $A_{\pattern}$ results from concatenating $A_{\subpattern}$ with $n-1$ many of its duplicates.
Note that although the semantics of concatenating $\subpattern$ with itself $n$ times and the repetition $\subpattern^{n \dots n}$ may appear similar, they are different. For instance, if $\subpattern$ is $() \xrightarrow{a} (x) \rightarrow (x) \rightarrow ()$, then the semantics of $\pattern_c = \subpattern\subpattern$ is the set of paths $v_0 \pathsep a_1 \pathsep v_1 \dots a_6 \pathsep v_6$ of length six, where $a_1 = a_4 = a$ and $v_1 = v_2 = \dots = v_5$. However, the semantics of $\pattern_r = \subpattern^{2 \dots 2}$ is the set of data paths $v'_0 \pathsep a'_1 \pathsep v'_1 \dots a'_6 \pathsep v'_6$ of length six, where $a'_1 = a'_4 = a$, $v'_1 = v'_2$, and $v'_4 = v'_5$.
Therefore, constructing $A_{\pattern}$ for $\pattern = \subpattern^{n \dots n}$ is analogous to concatenating $\subpattern$ with itself $n$ times, except that register unification is unnecessary. This construction is achieved simply by assigning each copy of $A_\subpattern$ its own independent set of registers.

Suppose $\pattern=\subpattern^{1 \dots \infty}$. Without loss of generality, we assume that there are no outgoing transitions starting from final states of $A_{\subpattern}$. In this case, $A_{\pattern}$ is derived from $A_{\subpattern}$, where: (1) for every transition $\delta$ of $A_{\subpattern}$ from state $s_1$ to state $s_2$, if state $s_2$ is final, then replacing $s_2$ with an initial state $s_0$ gives rise to a new transition $\delta'$, which is also part of $A_{\pattern}$; (2) all final states of $A_{\subpattern}$ are transformed into initial states of $A_{\pattern}$. Moreover, during the application of transition $\delta'$, all register values are set to $\sharp$.

Analogous to the union case, for every path $\rho$, if $\gamma$ is an accepting run of $A_{\pattern}$ on the data path $\tdp{\rho}$, then $(\rho,\mu) \in \semantics{\pattern}{G}$, where $\mu$ is derived from $\gamma$ as follows: the mapping $\mu$ assigns a free variable $x$ to $v$ if the registers corresponding to $x$ hold the value $v$, and to the null value $\nullnode$ if those registers remain unwritten along $\gamma$.

\paragraph*{Conditioning}
Suppose $\pattern = \subpattern_{\langle \theta \rangle}$, where $\theta = x \eqdata y$ for some node variables $x$ and $y$, and $A_{\subpattern}$ is the automaton for $\subpattern$. Without loss of generality, we assume there are no outgoing transitions starting from final states of $A_{\subpattern}$.

$A_{\pattern}$ is identical to $A_{\subpattern}$ except for those transitions ending in final state $s_f$. Suppose $\delta= s \xrightarrow{\precondt, \updatet} s_f$ is a transition of $A_{\subpattern}$. We create a new transition $\delta'$ which behaves as $\delta$ except that the precondition $\precondt$ is conjoined with the condition that the value of $r^{\datam}_x$ is equal to the value of $r^{\datam}_y$. In the case of $\theta = \theta_1 \wedge \theta_2$, $\theta_1 \vee \theta_2$, or $\neg \theta_1$, the construction of $A_{\pattern}$ follows a similar process.

\bigskip

This completes the inductive construction of $A_{\pattern}$. Moreover, we have the structural inductive invariant on patterns that for every path $\rho$, $\tdp{\rho}$ is accepted by $A_{\pattern}$ if and only if there exists a mapping $\mu$ such that $(\rho, \mu) \in \semantics{\pattern}{G}$.
Thus, (S) holds, and we conclude that for every \gpc query $Q$ of the form $\restrictor\ \pattern$, there exists a 1-ary \uedn formula $\phi_Q(\pi)$ equivalent to $Q$.

Thus we have completed the proof of (S).
%, and with it the proof of Proposition \ref{prop:gpctoaut}, which is the base case of Theorem \ref{thm:ednsubsumesgpc}.
As a result, we derive Theorem~\ref{thm:ednsubsumesgpc}.
\end{proof}

\subsection{The data complexity of \uedn}\label{sec:uedn_complexity}

In the previous subsection, we introduced a fragment \uedn of\linebreak \edn, which has no elementary bound on the data complexity, and showed that \gpc is subsumed by \uedn. Additionally, we asserted that \uedn has significantly lower data complexity than \edn, within \pspace in terms of the size of the data graph $G$ w.r.t. the standard encoding, which is basically polynomial in the number of the nodes and the number of edge labels. From this, we could conclude that \uedn is strictly less expressive than {\edn}. 
We will now substantiate this data complexity claim.

Each quantifier-free \edn formula $\psi$  is a Boolean combination of atoms of the form $\bar{\pi} \in A$ or $\bar{\pi} \not\in A$.
To check that a universally-quantified sentence does not hold on a given graph, we can negate it as an existentially-quantified Boolean combination, and then guess a conjunction of atoms of the above form that witness satisfiability.

Thus to show that \uedn query evaluation is in \pspace, it suffices to give a \pspace algorithm for checking satisfiability of
formulas:
\begin{equation} \label{eq:normal_form}
%\begin{align*}
\exists \bar{\omega}.  \bigwedge_i  \bar{\omega} \in A_i \wedge \bigwedge_j \neg \bar{\omega} \in B_j.  \tag{$\dagger$}
%end{align*}
\end{equation}
Since we are concerned with data complexity, the sizes of the \uedn formulas are treated as constants. Specifically, the number of variables, the number of {\ra}s $A_i$ and $B_j$, the number of conjunctions, and the sizes of $A_i$ and $B_j$ are all considered constant, where the size of an {\ra} refers to the number of states, transitions, and registers.

\begin{proposition} \label{prop:pspace}
Fix $A_i,B_i$, $i \leq n$,  $m$-ary {\ra}s and let $\phi:=\exists \bar{\omega}. \bigwedge_i \bar{\omega} \in A_i \wedge \bigwedge_j  \bar{\omega} \not\in B_j$.
There is an algorithm that checks if  $G \models \phi$ running in {\pspace} in the size of $G$.
\end{proposition}

\begin{proof}
For simplicity, we will assume $m=1$.  
Let $\mathcal{D}_G$ be the set of data values appearing in $G$ and $\Sigma_G = \mathcal{D}_G \cup \Sigma$, a set that is polynomial in $G$ for fixed $\Sigma$.
For any {\ra} $A$ there is a \ptime algorithm inputting $G$ and producing an NFA over $\Sigma_G$ that accepts exactly the data paths through $G$ that are accepted by $A$.
Intuitively, we can consider $G$ as an automaton, and the resulting automaton is the product automaton of $A$ and $G$, which can be constructed in quadratic time.
Applying this to each $A_i$ and $B_j$, we have reduced to checking whether a Boolean combination of NFA is non-empty. Since the size of the product automaton for a $B_j$ and $G$ is polynomial in the size of $G$, checking the emptiness of the complement of the product automaton can be done in \pspace with respect to the size of $G$.
\end{proof}

Combining the reduction to satisfiability of sentences of the form (Eq.~\ref{eq:normal_form}), we have shown that {\uedn} has much lower complexity than {\edn}:

\begin{corollary}
The data complexity of  \uedn query evaluation is in {\pspace}.
\end{corollary}

\subsection{\rd to \fotcd}

To show the inclusion of \rd in \fotcd, for every \ra $A$ we will construct an \fotcd formula $\phi_A$ simulating the computation of $A$. Without loss of generality, assume $A$ is a unary \ra over $\Sigma=\{a,b\}$
%\michael{I don't understand $A$ is unary and alphabet is binary} \dd{revised}
and it has $n=2^t$ states for some positive integer $t$.
Similar to the translation of an $n$-ary NFA to an \fotc formula for arbitrary $n$, the translation from an $n$-ary \ra $A$ to $\phi_A$ follows a process analogous to the 1-ary case. First, let's explore a simplified scenario of using an \fotc formula to emulate the computation of an NFA.

Let $G$ be a graph with only ids as data, containing at least two nodes $v$ and $w$. Let $A$ be an NFA.  We can assume $Q=\{\bm{r}_1,\dots,\bm{r}_n\}$, where each $\bm{r}_i$ is a $t$-tuple over $\{v,w\}$. We do this by taking $t$ large enough so that the size of $Q$ is at most $2^t$.
Since $Q$ is finite, we can define first-order formulas $\phi^{\tran}_a$ and $\phi^{\tran}_b$ such that (i) $\phi^{\tran}_a(\bm{r}_i,\bm{r}_j)$ holds if and only if $(\bm{r}_i,a,\bm{r}_j)$ is a transition in $A$ for $i,j \in [n]$ and (ii) $\phi^{\tran}_b(\bm{r}_i,\bm{r}_j)$ holds if and only if $(\bm{r}_i,b,\bm{r}_j)$ is a transition in $A$ for $i,j \in [n]$.

For any path $\rho$ in $G$, with label $\lb{\rho}=a_1\dots a_m \in L(A)$, there exist sequences of states $p_0 \dots p_m$ and nodes $v_0 \dots v_m$ in $G$ such that $p_0$ is the initial state, $p_m$ is a final state, $(p_{i-1},a_i,p_i)$ is a transition in $A$, and $(v_{i-1},a_i,v_i)$ forms an edge in $G$ for $i\in [m]$. We can define:
\[\psi^{\tran}_?(\bm{z},\bm{z}'):= \phi^{\tran}_?(\bm{x},\bm{x}') \wedge E_?(y,y'),\]
where $\bm{z}=(\bm{x},y)$ and $\bm{z}'=(\bm{x}',y')$ have the same arity, for $?=a,b$ and 
\[\psi_{\tran} := \psi^{\tran}_a \vee \psi^{\tran}_b.\] 
Thus $\lb{\rho} \in L(A)$ if and only if $\psi^*_{\tran}(\bm{r},\bm{s})$ holds where $\bm{r}=(\bm{r}_i,r)$, $\bm{s}=(\bm{r}_j,s)$, $\bm{r}_i$ is the initial state, $\bm{r}_j$ is a final state, and $\rho$ goes from node $r$ to node $s$.

Now, let's consider the general case where $A$ is an \ra and $G$ is a data graph. Suppose that $A$ has registers $\{1, \dots, \ell\}$ and $u = d_1 \pathsep d_2 \dots d_m$ is a data path in $G$. The data path $u$ is accepted by $A$ if there exists a sequence of accepting configurations $c_0 \pathsep c_1 \dots c_m$ and a sequence of transitions $\delta_1 \dots \delta_m$ such that $c_{i-1} \xrightarrow{d_i}_{\delta_i} c_i$ for each $i \in [m]$.
Assume $X_{\idm}=\{1,\dots,\eta-1\}$ and $X_{\datam}=\{\eta,\dots,\ell\}$ for some index $\eta$. A configuration of $A$ takes the form $(q,r_1,\dots,r_\ell)$, where $q$ is an $A$-state, $r_i=\idm(v_j)$ if $i < \eta$, and $r_i=\nodedata(v_j)$ if $i\geq \eta$ for some $v_j$ in $G$. The value of $r_i$ is uniquely determined once $v_j$ is known. %michael: stop with "Accordingly"
Thus $u$ is accepted by $A$ if and only if there exists a sequence $\alpha_0 \dots \alpha_m$, where $\alpha_i = (q_i,v^{(i)}_1,\dots,v^{(i)}_\ell)$ for some $A$-state $q_i$ and $G$-nodes $v^{(i)}_1,\dots,v^{(i)}_\ell$ such that $c_i = (q_i,\idm(v^{(i)}_1),\dots,\idm(v^{(i)}_{\eta-1}),\nodedata(v^{(i)}_\eta),\dots,\nodedata(v^{(i)}_\ell))$ for $i \in [m]$.

In the case of NFA emulation, we utilized tuples $\bm{r}$ of nodes to carry the state information of the automaton.  Likewise, for {\ra}s, we can also employ tuples of nodes to carry $\alpha_0, \dots, \alpha_m$ around.
Recall that there are two kinds of transitions in $A$: word transitions of the form $\delta_w = p \xrightarrow{a} q$ and data transitions of the form $\delta_d = r \xrightarrow{E,I,U} s$. When $\delta_w$ is applied, the content of the registers is not modified. Therefore the {\ra}-version formula $\psi^{\tran}_{\delta_w}(\bm{x}, y, \bm{z}, \bm{x}', y', \bm{z}')$ corresponding to $\delta_w$ can be defined as:
\begin{align*}
%\psi^{\tran}_w(\bm{x}, y, \bm{z}, \bm{x}', y', \bm{z}') := \\
\phi^{\tran}_a(\bm{x}, \bm{x}') \wedge E_a(y, y') \wedge \left(\bigwedge_{1 \leq i < \eta} (z_i = z'_i)\right) \wedge \left(\bigwedge_{\eta \leq i \leq \ell} (z_i \eqid z'_i)\right).
\end{align*}
When $\delta_{\delta_d}$ is applied, the values of the registers are compared with the data value of a node and updated. We can define the {\ra}-version formula $\psi^{\tran}_d$ corresponding to $\delta_d$ as follows:
\[
\psi^{\tran}_{\delta_d}(\bm{x}, y, \bm{z}, \bm{x}', y', \bm{z}', w) := \phi^{\tran}_d(\bm{x}, \bm{x}') \wedge \phi_{E}(\bm{z}, w) \wedge \phi_I(\bm{z}, w) \wedge \phi_U(\bm{z}, \bm{z'}, w),
\]
where:
\begin{itemize}
\item $\phi^{\tran}_d$ is the formula such that $\phi^{\tran}_d(\bm{r}, \bm{s})$ holds if and only if $\delta_d=r \xrightarrow{E, I, U} s$ is a data transition from $r$ to $s$, and states $r, s$ correspond to tuples $\bm{r}, \bm{s} \in Q$, respectively.
\item \[
\phi_E(\bm{z}, w) := \left(\bigwedge_{x_i \in E_{\idm}} z_i = w\right) \wedge \left(\bigwedge_{x_i \in E_{\datam}} z_i \eqdata w\right)
\]
states that the id of node $w$ is identical to the values of registers in $E_{\idm}$ and the data value excluding the id of node $w$ is identical to the values of registers in $E_{\datam}$.
\item \[
\phi_I(\bm{z}, w) := \left(\bigwedge_{x_i \in I_{\idm}} \neg(z_i = w)\right) \wedge \left(\bigwedge_{x_i \in I_{\datam}} \neg(z_i \eqdata w)\right)
\]
states that the id of node $w$ is different from the values of registers in $I_{\idm}$ and the data value excluding the id of node $w$ is different from the values of registers in $I_{\datam}$.
\item \[
\phi_U(\bm{z}, \bm{z}', w) := \left(\bigwedge_{x_i \not\in U} z'_i = z_i\right) \wedge \left(\bigwedge_{x_i \in U} z'_i = w\right)
\]
%\michael{Last item does not seem right: check at least $x_i$ vs $z_i$} \dd{revised}
states how the registers are updated.
\end{itemize}
Finally, let
\[
\psi_{\tran}(\bm{x},y,\bm{z},v,\bm{x}',y',\bm{z}',v'):= \bigvee_{\delta\in\Delta} \psi^{\tran}_{\delta} \wedge \phi_{\delta}(v,v'),
\]
where $\phi_\delta(v,v')$ holds if $\delta$ is a data transition with $(v,v')=(w,u)$, or if $\delta$ is a word transition with $(v,v')=(u,w)$. Thus, we have $\tdp{\rho} \in L(A)$ if and only if $\psi^*_{\tran}(\bm{x},y,\bm{z},w,\bm{x}',y',\bm{z}',u)$ holds, where $\bm{x}$ is the initial state, $\bm{x}'$ is a final state, $\bm{z}$ specifies that all registers hold the null value $\sharp$, and $\rho$ is a path from node $y$ to node $y'$. This completes the proof.

\subsection{Strictness of the inclusions in Figure~\ref{fig:existing}}

Since \gpcrd and \fotcr share the same expressive power \cite{LMV16}, \gpcrd is subsumed by full transitive closure logic, as shown in Figure~\ref{fig:existing}. The following result states that the subsumptions in the diagram are all strict.

\begin{theorem}\label{thm:prior}
Assuming \nlspace $\neq$ \np, no containment relations other than those depicted in the diagram are valid. Specifically, all the containment relations in the diagram are strict.
\end{theorem}

Since the diagram includes numerous containments, for the sake of clarity, we will initially enumerate a set of properties that we intend to derive.
\begin{enumerate}[P.1]
\item\label{prop:sublanguage_incomparable} Every pair from the following set is incomparable: \gpc, \wl, \cdn, \fotcd under the assumption that \nlspace $\neq$ \np.

\item\label{prop:edn_not_contain_gxpath} \gpcrd is not subsumed by \edn.
\item\label{prop:languages_not_contain_rd} \rd is not subsumed by \gpcrd, \gpc, or \wl.
\end{enumerate}

The assumption \nlspace $\neq$ \np will be used only in the first item, and there it will be used to derive that \gpc and \cdn are not subsumed by {\fotcd} and {\gpc} is not subsumed by {\cdn}.
In the previous section, we showed that \rd, a query language based on regular languages, is subsumed by \fotcd. However, we cannot extend this result to \gpc, although {\gpc} is also inspired by regular languages. That is, \gpc is not subsumed by \fotcd, due to the usage of restrictors in \gpc. Specifically, it is unclear whether the notion of ``simple'' can be expressed in \fotcd.

These properties are sufficient to obtain the following, which implies that no other containments are derivable:
\begin{enumerate}[R.1]
\item \fotcd and \edn are incomparable.
\item \gpc, \wl, \cdn are strictly subsumed by \edn. 
\item \rd is strictly subsumed by \cdn and \fotcd. 
\item \gpc, \wl, \rd, and \gpcrd are incomparable.
\item \gpcrd is incomparable with \cdn and \edn.
\item \gpcrd is strictly subsumed by \fotcd.
\end{enumerate}

\begin{figure}[t]
\centering
\includegraphics[width=0.8\textwidth]{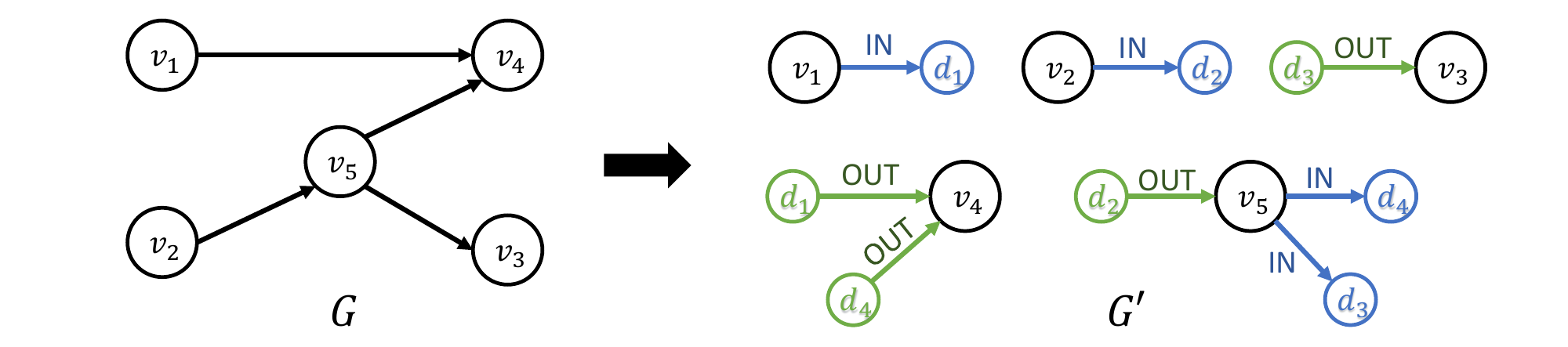}
\caption{Directed graph $G$ translated to data graph $G'$.}
\label{fig:Link}
\end{figure}

We start with P.\ref{prop:sublanguage_incomparable}.  
We begin by showing that {\fotcd} is not subsumed in any of the other languages considered in P.\ref{prop:sublanguage_incomparable}.
%michael: start by explaining what we wan to show before going into the encoding
To give intuition for the separating example, consider an encoding of an ordinary graph $G$, with arbitrary connectivity, as a highly disconnected
data graph $G'$, where in $G'$ each node is linked to only two other nodes.
%The notion of the ``connectivity'' of two nodes in a graph can be defined in \fotc, as well as in \fotcd. Thus, there exists an \fotcd formula $\connection$ such that for every data graph $G$ and two nodes $s, t$ in $G$, $\connection(s, t)$ holds if and only if there is a path from $s$ to $t$ in $G$. 
The translation is shown in Figure~\ref{fig:Link}. The edge connection between nodes $v_2$ and $v_4$ in $G$ is represented in $G'$ as an
$\outm$-labelled edge to a node carrying a certain data value, then an $\inm$-edge from another node with the same data value to $v_4$. 
%For example, $\connection(v_2, v_4)$ holds for the graph $G$, but not for $G'$.  
%Here, $G$ is a graph without data, encoded as a data graph \michael{Referee suggests we explain why we are introducing this encoding}  %$G'$ with edge labels $\outm$ (colored green) and $\inm$ (colored blue), as illustrated in the figure. The fact that nodes $v_2$ and %$v_4$ of the data graph $G'$ are connected can still be expressed in \fotcd.  
%Under this encoding, we can transform a directed graph $G$ into a data graph $G'$ such that any two original nodes from $G$ are disconnected in $G'$, regardless of whether a path connected them in $G$.
%\dd{An explanation is added.}
%\michael{I don't understand this explanation, or any of the text that was in this paragraph, so I wrote}

Based on this, we create our separating example. We formulate an \fo formula $\datalink(x, y)$ such that a pair of nodes $v$ and $w$ in $G'$ satisfies $\datalink$ if there exists a node $s$ linked to $v$ through an $\outm$ edge and another node $f$ linked to $w$ through an $\inm$ edge, with the data value of $s$ equating the data value of $f$. For instance, $(v_1, v_4)$ satisfies $\datalink$, while $(v_3, v_5)$ does not.  
Accordingly, we can create an \fotcd formula $\dataconnection(x, y)$ representing the transitive closure of $\datalink$.  
Consequently, a pair of nodes $(v, w)$ in $G'$ satisfies $\dataconnection$ if $v$ and $w$ are connected in $G$.

We now argue that this formula cannot be expressed in the other languages mentioned in  P.\ref{prop:sublanguage_incomparable},
%anguage
like {\edn}.  To see this, consider the family of graphs similar to $G'$ in Figure~\ref{fig:Link}, and note that path lengths are confined to $2$. Therefore, fixing a vocabulary with three relations $P_0, P_1, P_2$ over nodes -- where $P_0$, $P_1$, and $P_2$ are unary, binary, and ternary, respectively -- we can convert any \edn formula $f$ into a first-order logic formula without path variables. But connectivity cannot be expressed using first-order logic \cite{Gradel:92:CSL}. Consequently, \edn does not subsume \fotcd, and thus \fotcd is not subsumed by any of the other three languages.

It is known that there is no elementary bound on the data complexity for \wl, even for graphs with only ids as data \cite{BFL15}. But \fotcd, \gpc, and \cdn do have elementary bounds on data complexity, with  \fotcd having a {\nlspace} bound \cite{Immerman:88:SIAM}. Consequently, these languages cannot subsume \wl.

What remains in  P.\ref{prop:sublanguage_incomparable} is to show that \gpc and \cdn cannot be subsumed by  \wl, by \fotcd, or by each other. It is known \cite{MNP23}
that GPC can express NP-hard queries (in particular, $Q_{\even}$).
\begin{proposition}{(\cite{MNP23})}\label{prop:gpc_complexity}
 Query evaluation for \gpc (even without data) is $\np$-hard in data complexity.
 \end{proposition}
We now argue that
 \cdn are both \np-hard in data complexity. Since \fotcd is in \nlspace, this will imply that \fotcd cannot subsume either of these languages, assuming \np is not equal to {\nlspace}.
 %\begin{proof}. 
% \michael{Commented out the previous stuff, which was not a proof. Could someone give a pattern query that is $\np$ hard and explain why?}
 %\end{proof}
%michael: referee asked for this to be deleted: we have a formal proof below so I think it is ok
%he reduction of the Hamiltonian path problem to a \cdn query evaluation problem is done by creating $n$ copies $V_1, %\ldots, V_n$ of the graph's node set $V$, where each transition in the original graph $G$ from node $v$ to node $w$ is %replicated by adding a transition from the copy of $v$ in $V_i$ to the copy of $w$ in $V_{i+1}$ for $i = 1, \ldots, n-%1$. Thus, $G$ contains a Hamiltonian path if and only if there exists a path $\rho$ from a node in $V_1$ to a node in %$V_n$, with all nodes along the path corresponding to distinct nodes in $V$. This condition can be verified in \cdn due %to the availability of node ids and negation.

\begin{proposition}
Query evaluation for \cdn is \np-hard in terms of data complexity.
\end{proposition}
\begin{proof}
We reduce the Hamiltonian path problem to \cdn query evaluation.
Given a directed graph $G=(V,E)$ where $|V|=n_G$ and $V=[n_G]$, we let $G'=(V',E',\idm,\nodedata)$ be the data graph where
\begin{itemize}
\item $\Sigma=\{a,b\}$,
\item $V' = \bigcup_{i=1}^{n_G} V_i$ where $V_i = \{v + i\cdot n_G \in \mathbb{N} \mid v \in V\}$, 
%\michael{referee should $n$ be $n_G$ in the union indices (also several places in the bullet itme below)}
%\dd{Modified.}
\item $E'=\bigcup_{i=2}^{n_G} E_i$ where $E_i = \{(v+(i-1)\cdot n_G ,a_i,w+i\cdot n_G) \in V_{i-1}\times \Sigma \times V_i \mid (v,w) \in E\}$ and %$a_i = b$ when $i\not\in \{2,n_G\}$ while 
$a_i=a$ iff $i \in \{2, n_G\}$,
\item $\nodedata(v+i\cdot n_G) = v$ for each $v \in V$ and $1 \leq i \leq n_G$.
\end{itemize}
According to the construction of $G'$, we have that $G$ contains a Hamiltonian path if and only if there is a path $\pi$ in $G'$ such that the data path $\tdp{\pi} = d_1 \pathsep a_2 \pathsep d_2 \dots a_{n_G} \pathsep d_{n_G}$, $a_2\dots a_{n_G} \in ab^*a$, and $d_1\dots d_{n_G} \not\in L=\{e_1\dots e_m \in \mathbb{N}^* \mid e_i = e_j \textit{ for some } 1\leq i < j \leq m\}$. Regular languages are closed under complement, so there is a finite automaton recognizing the complement of $ab^*a$. It is known that $L$ is recognizable by a register automaton \cite{Kaminski:94:TCS}. Consequently, we have that there is an \ra $A$ such that the data path $\tdp{\pi} = d_1 \pathsep a_2 \pathsep d_2 \dots a_{n_G} \pathsep d_{n_G} \not\in L(A)$ if and only if $a_2\dots a_{n_G} \in ab^*a$ and $d_1\dots d_{n_G} \not\in L$.
\end{proof}

While \gpc and \cdn, which are based on regular languages, can express the query $Q_{\even}$, we know from \cite[Proposition 7]{Hellings:13:ICDT} that $Q_{\even}$ is not expressible in \wl for graphs with only ids as data.  Thus neither of these languages is subsumed by {\wl}.

To finish the derivation of Property P.\ref{prop:sublanguage_incomparable}, we need to show that neither of \cdn and \gpc subsumes the other. For one non-containment, we use that {\gpc} is NP-hard even without data (Proposition~\ref{prop:gpc_complexity}), while {\cdn} is in {\nlspace} without data \cite{Lin:12:DS}. For the other, we show that even \rd is not subsumed by {\gpc}: this is shown in the argument for Property P.\ref{prop:languages_not_contain_rd} below.

We turn to P.\ref{prop:edn_not_contain_gxpath}. Previously, we derived that \fotcd is not subsumed by \edn by showing that $\dataconnection$ cannot be expressed within \edn. This example can be used to show non-containment of \gpcrd within \edn. The languages \gpcrd and \fotcr are equivalent in expressiveness \cite{LMV16}. The language \\
\fotcr allows us to express $\dataconnection=\datalink^*$, where $\datalink(x,y)$ is 
\[
\exists z. (\outm(z,x) \wedge (\exists x. (\inm(x,y) \wedge x \eqdata z)))
\]
As a result, we have P.\ref{prop:edn_not_contain_gxpath}.

Next, we argue for  P.\ref{prop:languages_not_contain_rd}. Given that \rd is also based on regular languages, it can express the query $Q_{\even}$. Consequently, we deduce that \wl does not subsume \rd.
Let's define $Q_{\textsc{4nodes}}$ as the query: 
\begin{align*}
Q_{\textsc{4nodes}}:= & \textit{``Is there a path from node $s$ to node $t$}\\
& \textit{containing four distinct nodes?''.}
\end{align*} 
For any complete graph $G$ with two nodes $s$ and $t$, there exists a path from $s$ to $t$ containing four distinct nodes if and only if $G$ has at least four nodes.
It was established that no \gpcrd formula can differentiate between $K_3$ and $K_4$ \cite{LMV16}, where $K_n$ is the complete graph on $n$ nodes for all $n\in \mathbf{Z}_{>0}$. More specifically, there is no \gpcrd sentence $\phi$ such that $K_4 \models \phi$, but $K_3 \not\models \phi$.
Therefore, $Q_{\textsc{4nodes}}$ cannot be expressed using the language \gpcrd. However, the query $Q_{\textsc{4nodes}}$ is expressible using \rd. So we conclude that \gpcrd does not subsume \rd.

\bigskip

Lastly, to show that \rd is not subsumed by \gpc, we require the notion of \emph{normal form patterns} and their corresponding properties:

%\paragraph*{Normal forms}
\begin{definition}
A \gpc pattern $\pattern$ is in \emph{normal form} if $\pattern = \pattern_1 + \dots + \pattern_n$, where each $\pattern_i$ is of the form $\{{\primepattern_1} \dots {\primepattern_m}\}_{\langle \theta \rangle}$ such that for every $j \in [m]$, $\primepattern_{j}$ is either:
\begin{itemize}
    \item an atomic pattern, or
    \item $\primepattern_{j}= \primepattern^* (= \primepattern^{0..\infty})$ for some pattern $\primepattern$.
\end{itemize}
%is recursively defined in one of the following forms:
%\begin{align*}
%\mbox{an atomic pattern} \\
%\primepattern \primepattern'  \\
%\primepattern_{\langle \theta \rangle} \\
%\primepattern^* (= \primepattern^{0..\infty})
%\end{align*}
%\michael{I don't follow this: is it a recursive definition?}
%\dd{No. All we need is that the condition $\theta$ for free variables is outermost and the repeat pattern is $\pattern^*$.}
%\michael{I am lost then. Could you just spell out in words what you mean, since formally it doesn't make sense to me.  What is the third case? Is $\pattern_{\langle \theta \rangle} \pattern$ allowed?  And what is being excluded here?}
%\dd{Intuitively, the normal form is like the disjunct normal form (DNF).} ok that helps
\end{definition}

By applying the following two families of rules, we can normalize every \gpc pattern. The first family pushes filters out of compositions:

\begin{proposition}
$\pattern_{\langle \theta \rangle} \primepattern_{\langle \eta \rangle} \approx \{\pattern\primepattern\}_{\langle \theta \wedge \eta \rangle}$ 
if the pattern $\pattern_{\langle \theta \rangle} \primepattern_{\langle \eta\rangle}$ is valid.
\end{proposition}
\begin{proof}
Suppose $(\rho,\mu) \in \semantics{\pattern_{\langle \theta \rangle} \primepattern_{\langle \eta \rangle}}{G}$.  
Then, we have $(\rho,\mu) \in \semantics{\pattern \primepattern}{G}$, $\mu \models \theta$, and $\mu \models \eta$.  
These imply $(\rho,\mu) \in \semantics{\pattern \primepattern}{G}$ and $\mu \models \theta \wedge \eta$.  
Accordingly, we have $(\rho,\mu) \in \semantics{\{\pattern\primepattern\}_{\langle \theta \wedge \eta \rangle}}{G}$.  

Conversely, if $(\rho,\mu) \in \semantics{\{\pattern\primepattern\}_{\langle \theta \wedge \eta \rangle}}{G}$, then  
$(\rho,\mu) \in \semantics{\pattern \primepattern}{G}$ and $\mu \models \theta \wedge \eta$, which implies  
$\mu \models \theta$ and $\mu \models \eta$.  
Thus, we conclude $(\rho,\mu) \in \semantics{\pattern_{\langle \theta \rangle} \primepattern_{\langle \eta \rangle}}{G}$.
\end{proof}

A second family of normalization rules eliminates filters from Kleene star patterns:

\begin{proposition}
$\{\pattern^*\}_{\langle \theta \rangle} \approx \{\pattern^*\}_{\langle \top \rangle} \approx \pattern^*$ if $\{\pattern^*\}_{\langle \theta \rangle}$ is valid.
\end{proposition}
\begin{proof}
By the definition of \gpc patterns, we have $\fv{\pattern^*} = \emptyset$.  
Accordingly, if $\{\pattern^*\}_{\langle \theta \rangle}$ is valid, then $\theta$ must be equivalent to $\top$.
\end{proof}

Note also that $\pattern^{n...m}$ can be eliminated, since it can be seen as $\pattern^{n...n}+...+\pattern^{m...m}$, and each $\pattern^{i...i}$ is equivalent to the concatenation of $i$ copies of $\pattern$, with each copy using a fresh set of node variables. This allows us to consider only the case $\pattern^{0...\infty}$.

%\michael{Are we now ready to prove the normal form lemma?}
%\dd{By the propositions, we immediately obtain the lemma.}
%\dd{The paragraphs are reordered.} michael: ok

As a result, we have:

\begin{lemma}
Every \gpc pattern is equivalent to a normal form pattern.
\end{lemma}

%\michael{So we are at the end of the proof, and the comment below is something different? It is not immediate to me how the propositions prove the lemma, although that is probably since I (like many readers) have not completely absorbed all the definitions}

We will also rely on the fact that patterns cannot distinguish between two graphs that are isomorphic:

\begin{proposition}\label{prop:isomorphic_graphs}
Let $\pattern$ be a \gpc pattern. For any two data graphs 
$G_1 = (V_1, E_1, \nodedata_1)$ and $G_2 = (V_2, E_2, \nodedata_2)$,  
if $V_1 = V_2$, $E_1 = E_2$, and  
$\nodedata_1(v_i) = \nodedata_1(v_j) \Leftrightarrow \nodedata_2(v_i) = \nodedata_2(v_j)$ for all $i, j$,  
and if $\mu_1, \mu_2$ are two mappings such that  
$\nodedata_1(\mu_1(x)) = \nodedata_2(\mu_2(x))$ for all $x \in \fv{\pattern}$,  
then  
$
(\rho_{G_1}, \mu_1) \in \semantics{\pattern}{G_1} \Leftrightarrow (\rho_{G_2}, \mu_2) \in \semantics{\pattern}{G_2}.
$
\end{proposition}

Now we are ready to show that \rd is not subsumed by \gpc.

\begin{theorem}\label{thm:GPC_not_contain_RDPQ}
\rd is not subsumed by \gpc.
\end{theorem}
\begin{proof}
To prove this theorem, we will show that the following query, which is evidently expressible in \rd, cannot be expressed in \gpc:
%define a query that is expressible in \rd and demonstrate that this query cannot be expressed in \gpc. Our query is:
\begin{align*}
Q_{\threeval}(s,t):= & \textit{``Is there a path $\rho$ from node $s$ to node $t$ where the number of}\\ 
& \textit{distinct data values (excluding ids) along $\rho$ is exactly 3?''.}
\end{align*}
%It is evident that $Q_{\threeval}$ can be expressed in \rd.

%Note that a \gpc query must be a join of  \gpc queries $Q_1,\dots,Q_n$  where $Q_i= \restrictor_i\ \pattern_i$ for some restrictor $\restrictor_i$ and pattern $\pattern_i$ for each $i$. For each data graph $G$ and $n$-tuple $\bm{s}=((\rho_1,\mu_1),\dots,(\rho_n,\mu_n))$, the tuple $\bm{s}\in \semantics{Q}{G}$ if and only if $(\rho_i,\mu_i)\in \semantics{Q_i}{G}$ for all $i,j$ and the union $\mu_1\cup \dots \cup \mu_n$ is defined. Suppose that $(\rho_i,\mu_i)\in \semantics{Q_i}{G}$ and $\tdp{\rho_i}$ contains exactly 3 distinct data values. Because the satisfaction of $Q_i$ depends only on $(\rho_i,\mu_i)$, removing the other components $Q_j$ ($j \neq i$) from the join does not affect this membership. 

%We will show that there is no \gpc path query of the form $\restrictor\ \pattern_{\threeval}$ defining the set of paths containing exactly three distinct data values from $s$ to $t$. 
%\michael{referee asks why}. 
%\dd{More details are given.}

\begin{figure}[h]
\centering
\begin{tikzpicture}
\centering
[node distance=160pt]
\tikzstyle{state}=[draw,shape=circle,minimum size=20pt]

\tikzset{invisible/.style={minimum width=0mm,inner sep=0mm,outer sep=0mm}}

\node[state] (v0) at (0,0) {$s$};
\node[state] (v1) at (1.5,0) {};
\node[state] (v2) at (3,0) {};
\node[state] (vm) at (6,0) {};
\node[state] (vn) at (7.5,0) {$t$};

%p
\path
(v0)
    edge [->] node[above] {} (v1)
(v1)
    edge [->] node[above] {} (v2)
(v2)
    edge [->,dashed] node[above] {} (vm)
(vm)
    edge [->] node[above] {} (vn)

;

\end{tikzpicture}
 
\caption{Chain graph $G$ and path $\rho_G$ from $s$ to $t$.}
\label{fig:chain_2}
\end{figure}

We will actually show a stronger separation: $Q_{\threeval}$ cannot be expressed in $\gpc$ even over chain graphs.
Suppose we restrict our attention to the chain data graphs shown in Figure~\ref{fig:chain_2}, and we use $\rho_G$ to denote the unique, shortest path from $s$ to $t$ in graph $G$. In such graphs, a path from $s$ to $t$ contains exactly three distinct data values if and only if it is $\rho_G$. Accordingly, we can restrict ourselves to the case where $\restrictor=\shortest$. This assertion holds even when restricting to chain data graphs over a unary alphabet $\Sigma$ (i.e., $\Sigma$ contains exactly one symbol).

\begin{figure}[t]
\begin{subfigure}[t]{0.49\textwidth}
\centering
\includegraphics[width=1\textwidth]{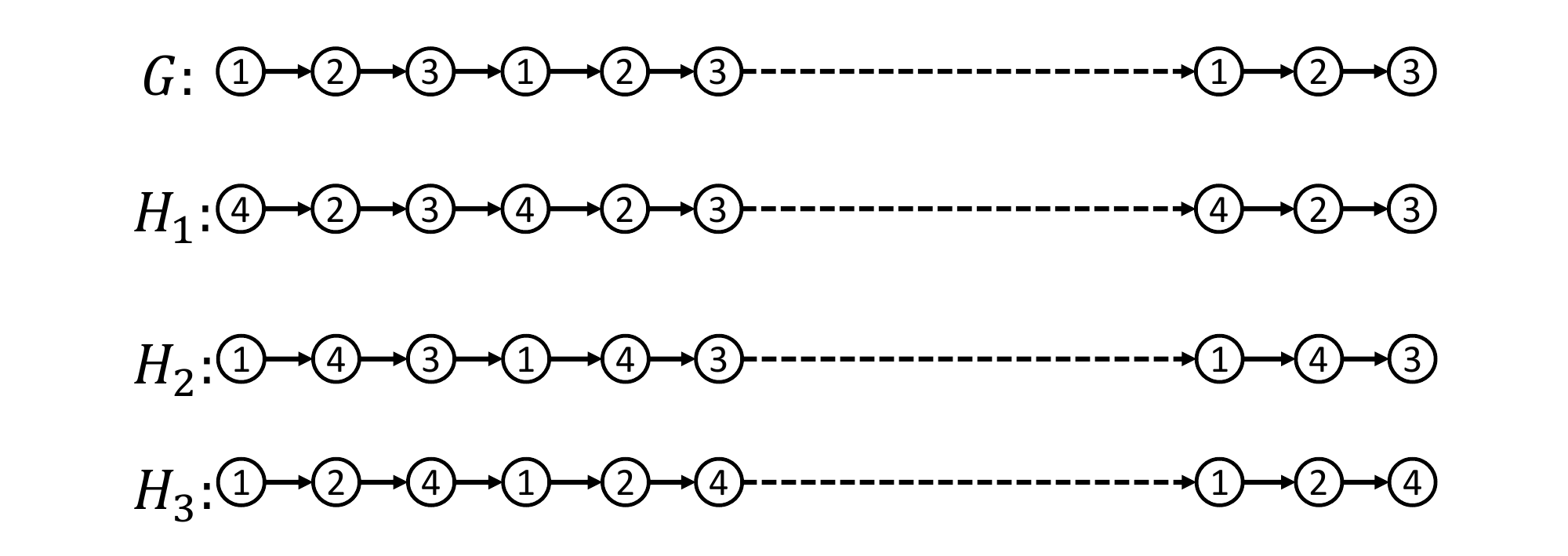}
\caption{Three isomorphic copies $H_1,H_2, H_3$ of $G$.}
\label{fig:G_to_H}
\end{subfigure}\hfill
\begin{subfigure}[t]{0.49\textwidth}
\centering
\includegraphics[width=1\textwidth]{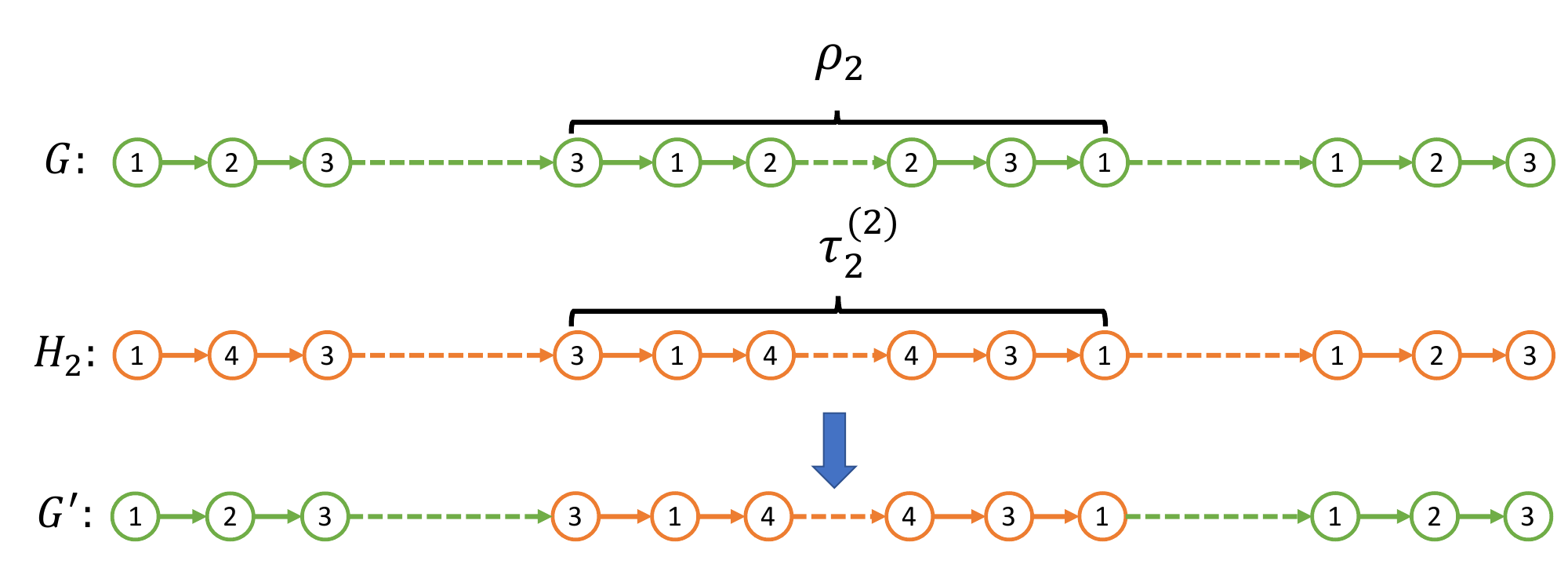}
\caption{Data graph $G'$ derived from merging $G$ and $H_2$.}
\label{fig:mergedgraph}
\end{subfigure}\hfill
\caption{A data graph $G$ with exactly 3 distinct data values (excluding ids).}
\end{figure}

Before continuing, we provide an overview and the main idea of the proof showing why a pattern equivalent to $Q_{\threeval}$ does not exist. %\michael{what result?}\dd{revised} 
We will use a proof by contradiction. Specifically, we will demonstrate that for any pattern $\pattern$, it is possible to construct two graphs, $G$ and $G'$, such that $\rho_G$ and $\rho_{G'}$ either both agree with $\pattern$ or both disagree with $\pattern$. However, $G$ contains exactly three distinct data values, while $G'$ does not.
See Figure~\ref{fig:G_to_H}. The data graph $G$ contains exactly three distinct data values (excluding ids), namely 1, 2, and 3. Additionally, the data graphs $H_1, H_2, H_3$ are isomorphic to $G$ and to each other. By Proposition~\ref{prop:isomorphic_graphs}, we have that for any pattern $\pattern$, $\rho_{G}$ is in the semantics of $\pattern$ if and only if $\rho_{H_i}$ is in the semantics of $\pattern$ for $i \in \{1, 2, 3\}$.
See Figure~\ref{fig:mergedgraph}. Furthermore, we can construct a data graph $G'$ from $G$ and $H_i$ for some $i \in \{1,2,3\}$ such that $G'$ contains exactly four distinct data values (excluding ids) and satisfies $\pattern$ in the same way as $G$, $H_1$, $H_2$, and $H_3$.

%Refer to Figure~\ref{fig:chain_2}, where we use $\rho_G$ to denote the longest path, i.e., the path from $s$ to $t$, in graph $G$. We will prove that there does not exist a \gpc pattern $\pattern_{\threeval}$ such that, for each input graph $G$, $\rho_G$ matches pattern $\pattern_{\threeval}$ if and only if $\rho_G$ is an answer to query $Q_{\threeval}$.

\bigskip

We now proceed to provide the detailed proof. Suppose that $\pattern_{\threeval}$ is a \gpc pattern such that $\restrictor\ \pattern_{\threeval}$ expresses $Q_{\threeval}$, where $\restrictor=\shortest$.

For simplicity, throughout the rest of this section, when comparing two chains $G_1=(V_1,E_1,\idm_i,\nodedata_1)$ and $G_2=(V_2,E_2,\idm_2,\nodedata_2)$, if $|V_1| \leq |V_2|$, we assume $V_1 \subseteq V_2$ and $E_1 \subseteq E_2$, and we will directly ignore the functions $\idm_1$ and $\idm_2$, as well as the labels of edges, treating paths as sequences of nodes.
Furthermore, given a chain $G=(V,E,\nodedata)$, we assume that $V=\{v_0,\dots,v_n\}$ for some $n$, and $E=\{(v_0,v_1),\dots,(v_{n-1},v_n)\}$. When referring to a path $\rho=v_i\dots v_j$ in $G$ for some $i\leq j$, we use $G_{\rho}$ to denote the subgraph $(V_{i,j},E_{i,j},\nodedata)$ from node $v_i$ to $v_j$, where $V_{i,j}=\{v_i,\dots,v_j\}$ and $E_{i,j}=\{(v_i,v_{j-1}),\dots,(v_{i+1},v_j)\}$.
For any two patterns $\pattern$ and $\primepattern$, we write $\pattern \approx \primepattern$ if $\semantics{\pattern}{G} = \semantics{\primepattern}{G}$ for all $G$.

If $\pattern_{\threeval}$ is a disjunction of two sub-patterns $\pattern_1$ and $\pattern_2$, and a path $\rho_G$ of chain graph $G$ matches the pattern $\pattern_{\threeval}$, then $\rho_G$ matches $\pattern_1$ or $\pattern_2$. 
Accordingly, by the normal form assumption, we can assume that $\pattern_{\threeval}$ is equivalent to $\{\subpattern_1 \dots \subpattern_m\}_{\langle \theta \rangle}$ for some condition $\theta$ and patterns $\subpattern_1, \dots, \subpattern_m$, where each $\subpattern_j$ is either an atomic pattern or has the form $\subsubpattern^*$ for some pattern $\subsubpattern$.

Now, consider any chain graph $G$ with three alternating values that is sufficiently long, denoted as $G=(V,E,\nodedata)$, where 
\[|V|=3^\ell> 3\cdot|\pattern_{\threeval}|+7\]
for some positive integer $\ell$. Here, $\nodedata(\rho_G)=(d_1 \pathsep d_2 \pathsep d_3)^\ell$, and $|\{d_1,d_2,d_3\}|=3$. We have $(\rho_G,\mu) \in \semantics{\pattern_{\threeval}}{G}$ for some $\mu$. 
Since $\frac{|V|-7}{3}> |\pattern_{\threeval}|$, there is an index $j$ and a pattern $\subsubpattern$ such that $\subpattern_j = \subsubpattern^*$: in other words, there is some segment of the graph matched by a wildcard sub-pattern.

Consider $\theta = \top$. There are two possibilities: $m=1$ and $m>1$.
The case $m=1$ can be easily transformed into case $m>1$ because for every pattern $\subsubpattern$, $\subsubpattern^*$ is equivalent to $\subsubpattern^*\subsubpattern^*$. 
Therefore, we will assume $m>1$: the composition has more than one component.

Up to this point, we have been reasoning about the pattern $\pattern_{\threeval}$ on the chain graph $G$. Now we define several additional isomorphic copies of $G$. Let $e \not\in \{d_1,d_2,d_3\}$, and define $H_i=(V,E,\nodedata_s)$ for $i=1,2,3$, where $\nodedata_1(\rho_{H_1})=(e \pathsep d_2 \pathsep d_3)^\ell$, $\nodedata_2(\rho_{H_2})=(d_1 \pathsep e \pathsep d_3)^\ell$, and $\nodedata_3(\rho_{H_3})=(d_1 \pathsep d_2 \pathsep e)^\ell$. For instance, let $d_1=1, d_2=2, d_3=3$, and $e=4$. Then, $G, H_1, H_2, H_3$ are illustrated in Figure~\ref{fig:G_to_H}. 

Since all four graphs are isomorphic, and the set of nodes and the set of edges of $H_i$ are $V$ and $E$, respectively, for all $i=1,2,3$, by Proposition~\ref{prop:isomorphic_graphs}, we have $(\rho_{H_i},\mu) \in \semantics{\pattern_{\threeval}}{H_i}$ for $i=1,2,3$.
By our assumption, $\pattern_{\threeval}=\{\subpattern_1\dots \subpattern_m\}_{\langle \theta \rangle}$, the path $\rho_G$ (and $\rho_{H_i}$ for $i=1,2,3$) can be divided into segments $\rho_1,\dots,\rho_m$ (and $\tau^i_1,\dots,\tau^i_m$ for $i=1,2,3$), where $(\rho_j,\mu_j) \in [\![ \subpattern_j ]\!]_{G}$ (and $(\tau^i_j,\mu_j) \in [\![ \subpattern_j ]\!]_{H_i}$ for $i=1,2,3$) for $j=1,\dots,m$, where $\mu_j$ is $\mu$ restricted to the variables of $\primepattern_j$. Additionally, we can assume that $|\rho_j|=|\tau^i_j|$ for $i=1,2,3$ for all $j$.

Since $|\rho_G|=|\rho_{H_1}|=|\rho_{H_2}|=|\rho_{H_3}|=|V| > 3\cdot|\pattern_{\threeval}|+7$, there is an index $1\leq \iota \leq m$ such that $3 \leq |\rho_\iota|=|\tau^i_\iota|$ for $i=1,2,3$ and $\subpattern_\iota = \subsubpattern^*$ for some $\subsubpattern$. 
Since $\subsubpattern^* \subsubpattern^*$ is equivalent to $\subsubpattern^*$, we can also demand that $|\rho_\iota| < |\rho_{G}| - 3$. In other words, there is a ``non-trivial'' segment matched by a Kleene star pattern component within $G$, and likewise within each of the $H_i$. This is a segment where we will perform our merging.

Let $3 < a < b < 3^\ell-3$ be the indices marking the beginning and end of the crucial segment, corresponding to nodes $v_a$ and $v_b$ in graph $G$, such that $\rho_\iota=\rho_G[v_a,v_b]$.
Note that within each $H_i$, the positions that contain the value $4$ are congruent modulo $3$. As we consider $H_1, H_2, H_3$, the positions of the value $4$ cover every possible residue class modulo $3$. Therefore, in one of the $H_i$, which we will call $H_r$, the positions of nodes with value $4$ do not align with $a$ or $b$ modulo $3$.
For this particular $r$, we have agreement between $G$ and $H_r$ at the endpoints of the crucial segment: $\nodedata(v_a)=\nodedata_r(v_a)$ and $\nodedata(v_b)=\nodedata_r(v_b)$. For example, if $\nodedata(v_a)=3$ and $\nodedata(v_b)=1$, then we have $r=2$.

Let $G'=(V,E,\nodedata')$ where:
\begin{itemize}
\item $\nodedata'(v_r)=\nodedata(v_r)$ if $r < a$ or $b<r$;
\item $\nodedata'(v_r)=\nodedata_r(v_u)$ if $a\leq u \leq b$.
\end{itemize}
Figure~\ref{fig:mergedgraph} illustrates the process of merging $G$ with $H_r$ to obtain $G'$.

We now claim that $(\rho_{G'},\mu) \in \semantics{\pattern_{\threeval}}{G'}$. Formally, we divide $\rho_{G'}$ into $\rho'_1,\dots,\rho'_m$ such that $|\rho'_p| = |\rho_p|$ for all $p=1,\dots,m$. Then we have $\rho'_\iota = \tau^r_\iota$, $\nodedata'(\rho'_\iota)=\nodedata_r(\tau^r_\iota)$, $\rho'_p=\rho_p$, and $\nodedata'(\rho'_{p'})=\nodedata(\rho_p)$ for each $p \in \{1,\dots,m\}\setminus \{\iota\}$. Thus $(\rho'_p,\mu) \in [\![ \subpattern_p ]\!]_{G'}$ for $p=1,\dots,m$, and $(\rho_{G'},\mu) \in  [\![ \subpattern_1\dots \subpattern_m ]\!]_{G'} \subseteq [\![ \pattern_{\threeval} ]\!]_{G'}$. Since $3 < |\rho'_\iota| < |V|-3$, we have $\nodedata'(\rho_{G'})=4>3$. This leads to a contradiction.

The analysis for the case $\theta \neq \top$ is similar. For instance, suppose $x \eqdata y$ is a sub-formula of $\theta$. Then, the node variables $x$ and $y$ do not appear in $\subpattern_j = \subsubpattern^*$ since $\fv{\subsubpattern^*}=\emptyset$. Therefore, the existence of $x$ and $y$ does not affect our construction of $G'$. This reasoning applies to all variables mentioned in $\theta$. Hence, Theorem~\ref{thm:GPC_not_contain_RDPQ} is proved.
\end{proof}

\section{Unifying the Languages}\label{sec:unifying}

A natural way to find a language subsuming all the languages in Figure~\ref{fig:existing} is to extend \edn or \fotcd, which are shown to be the maximally expressive query languages within the diagram in the preceding section. Essentially, \edn is a first-order logic. The extension of \edn augmented with transitive closures then intuitively subsumes \fotcd. However, there is an issue we need to address. In \edn, there are two types of variables: node variables and path variables. Given a formula $\phi(\bm{x},\bm{y})$, where $\bm{x}$ and $\bm{y}$ have the same arity, we cannot form $\phi^*(\bm{x},\bm{y})$ without considering the types of variables in $\bm{x}$ and $\bm{y}$. For instance, it does not make sense to allow a formula $\psi^*(x,\pi)$ from a formula $\psi(x,\pi)$ where $x$ is a node variable and $\pi$ is a path variable. We address this by imposing the restriction that all the variables in vectors $\bm{x}, \bm{y}$ of formula $\phi$ are node variables. 
%Second, we want the  language to remain well-behaved, i.e., the query evaluation problem should be at least decidable. We will define a query language \ednt augmenting \edn with transitive closure in such a way that it addresses the second postulate. 
%\michael{I do not understand "second postulate". Can we use something more down-to-earth?}
%\dd{The second postulate is deleted. BTW, it is a remnant of the most general extension of WL that we had in the previous version, which is undecidable.}

\begin{definition}
The formulas of \ednt are defined inductively by the following rules:
\begin{itemize}
\item \edn formulas are \ednt formulas.

\item If $\phi$ and $\psi$ are \ednt formulas, then $\neg\phi$, $\phi \vee \psi$, $\exists x. \phi$, and $\exists \pi. \psi$ are also \ednt formulas, where $x$ and $\pi$ are node and path variables appearing free in $\phi$ and $\psi$, respectively.

\item If $\phi(\bm{x},\bm{y})$ is an \ednt formula where $\bm{x},\bm{y}$ are vectors of node variables with the same arity, then $\phi^*(\bm{x},\bm{y})$ is an \ednt formula.
\end{itemize}
\end{definition}

Since \ednt is an extension of \edn, to define the semantics, we simply have to illustrate the meaning of $\phi^*(\bm{x},\bm{y})$, where $\bm{x}, \bm{y}$ are vectors of node variables with the same arity. Given a data graph $G$, we write $G \models \phi^*(\bm{v},\bm{v}')$ if there exists a sequence of vectors of nodes $\bm{v}_1=\bm{v},\dots,\bm{v}_n=\bm{v}'$ such that $G \models \phi(\bm{v}_i,\bm{v}_{i+1})$ for $1\leq i < n$.

By the definition of \ednt, we have the following theorem:

\begin{theorem}
\ednt subsumes all languages from Figure~\ref{fig:existing}.
\end{theorem}

Next, we show that \ednt is well-behaved.

\begin{proposition} \label{prop:edntdecide}
The \ednt query evaluation problem is decidable.
\end{proposition}
\begin{proof}
Let $\phi(\bm{x},\bm{\pi})$ be an \ednt formula and $G$ be a data graph. Without loss of generality, we assume that both $\mathcal{D}_{\idm},\mathcal{D}_{\propm} \subseteq \mathbb{N}$. Therefore, all data paths in $G$ are in the set $\mathbb{P}=\mathbb{N}^2\times (\Sigma \times \mathbb{N}^2)^*$.
The \ednt query evaluation is decidable if and only if the emptiness problem for the set $L_{\phi,G}=\{(\bm{v},\bm{\rho}) \mid G\models \phi(\bm{v},\bm{\rho})\}$ is decidable for arbitrary $\phi$ and $G$. We can decide emptiness using the following:

\begin{claim} \label{claim:regular} A 2-ary finite automaton $A_{\phi,G}$ recognizing $L_{\phi,G}$ can be effectively derived from $\phi$ and $G$.
\end{claim}

In \cite{Lin:12:DS}, it is shown that the statement holds when $\phi$ is an \edn formula and $G$ is a graph without data; thus, all automata $A$ appearing in $\phi$ are finite automata. We explain how to adapt the proof for data graphs. Here is the idea: consider the formulas $\phi = \exists \pi. \pi \in A$ and $\phi' = \exists \pi. \neg(\pi \in A)$. Given a graph $G$, we treat it as an automaton where all states are both initial and final. Recall that an NFA is a directed graph with an initial state and a set of final states.
%\michael{How?} \dd{explained}
The automaton $A_{\phi,G}$ is then the product automaton of $A$ and $G$. Consequently, $G \models \phi$ if and only if $A_{\phi,G}$ is non-empty. On the other hand, the automaton $A_{\phi',G}$ for the formula $\phi'$ is the product automaton of $A_\neg$ and $G$, where $A_\neg$ is the automaton recognizing the complement of $L(A)$. Moreover, $G \models \phi'$ if and only if $A_{\phi',G}$ is non-empty.

In \cite{Lin:12:DS}, the authors give an inductive construction of $A_{\phi,G}$. The main inductive step relies on closure of finite automata under complement. Although we consider data graphs, the automata $A$ appearing in $\phi$ are \ra{s}, and \ra{s} are not closed under complementation. However, once a data graph $G$ is fixed, we can transform all \ra{s} $A$ into finite automata $A'$ and derive an \edn formula $\phi'$ from $\phi$ by replacing $A$ with $A'$ such that $G \models \phi(\bm{v}, \bm{\rho})$ if and only if $G \models \phi'(\bm{v}, \bm{\rho})$.

For instance, consider Example~\ref{ex:edn_Hamiltonian}. 
There is no \ra recognizing the complement of $L(A_{\repeatm})$ in $\mathbb{P}$.
However, once $G$ is fixed, there are two finite subsets $D^G_{\idm},D^G_{\propm}$ of $\mathbb{N}$ s.t. all data paths in $G$ are over the union of $D^G_{\idm}$, $D^G_{\propm}$, and $\Sigma$. Thus there is a finite automaton $A^{G}_{\repeatm}$ s.t. for every data path $\rho$ in $G$, $\rho \in L(A_{\repeatm})$ if and only if $\rho \in L(A^G_{\repeatm})$.  
Hence, the claim is true when $\phi$ is an \edn formula, even when $G$ is with data.

From Claim \ref{claim:regular} Proposition \ref{prop:edntdecide} follows.
\end{proof}

Figure~\ref{fig:summary} summarizes the results of this section.

\section{Multi-Path Walk Logic, an extension of \wl}\label{sec:mwl}

In the previous section, we extended \edn to unify the prior languages. Based on Theorem~\ref{thm:prior}, we have that \edn subsumes \wl, \gpc, \rd, and \cdn. The query languages \rd for data graphs and \edn for graphs without data are well-studied \cite{LMV16, Lin:12:DS}. In this section, we examine \wl and derive an extension of \wl, still subsumed by {\edn}.

Consider a path $\rho = v_0a_1v_1 \dots v_n$. Node $v_i$ is at the $i^{th}$ position in $\rho$ for $i\in [0,n]$. In other words, every node in $\rho$ corresponds to a number in $[0,n]$. In \wl, we can compare nodes' data values through their positions. Additionally, we can compare node positions as numbers if they are on the same path. However, comparison across different paths is not allowed. That is, we can assert $l^\pi < n^\omega$ only if $\pi = \omega$. We call the extension of \wl obtained by lifting this limitation \emph{multi-path walk logic} (\mwl). Note that $\ell^\pi < n^\omega$ is expressible in \wl for different $\pi, \omega$ if $\ell$ and $n$ refer to some given constants.
In this section, we study the properties of \mwl. We relate it to the languages investigated in the preceding sections and derive the decidability of query evaluation.

%\subsection{Properties of \mwl}

Recall from Section \ref{sec:prior} that to complete the proof of Theorem~\ref{thm:existingcontainments}, we need to show that \wl is contained in \edn. We will show something stronger: that \mwl, an extension of \wl, is contained in \edn.

\begin{proposition}
\mwl is strictly stronger than \wl in expressive power.
\end{proposition}
\begin{proof}
%\dd{The proof has been rewritten.}

%Let $G$ be a graph, and let $n_{\red}$ and $n_{\black}$ be two nodes in $G$ with no outgoing edges. Assume that for every node $n$ in $G$, if $n \neq n_c$ for $c \in \{\red, \black\}$, then there is exactly one edge from $n$ either to $n_{\red}$ or to $n_{\black}$. We say $n$ is red if there is an edge from $n$ to $n_{\red}$, and black if there is an edge from $n$ to $n_{\black}$. Without loss of generality, assume that all labels of the edges to $n_{\red}$ are $a$, those to $n_{\black}$ are $b$, and that $a$ and $b$ are used nowhere else.

%A node is red if the following \wl formula is satisfied:

It is known that the query ``Are there two paths of different lengths from a given start node (represented as a singleton path) to a given target node (another singleton path)?'' is not expressible in \wl \cite[Proposition 7]{Hellings:13:ICDT}.

It is easy to see that this query is expressible as an {\mwl} query $Q^{\difm}_{\lenm}(\pi_1, \pi_2)$.
%\[
%\phi_{\red}(x) := \exists \pi, \ell, m. \, E_a(\ell, m) \wedge \ell \eqid x
%\]
%indicating that $x$ is connected to some node through an edge labeled $a$. The formula $\phi_{\black}$ is defined similarly, using a $b$ edge.
Below we give the idea, omitting for brevity the requirement that $\pi_1$ and $\pi_2$ are singletons (and thus represent single nodes):

\[ \exists \pi, \omega, \ell^\pi, m^\pi, s^\omega, t^\omega. (\first(\ell) \wedge \first(s) \wedge \last(m) \wedge \last(t)) \wedge \ell < m \]
\[\wedge \phi_{begins}(\pi, \pi_1) \wedge \phi_{begins}(\omega, \pi_1) \wedge \phi_{ends}(\pi, \pi_2) \wedge \phi_{ends}(\omega, \pi_2), \]
where:
\[
\first(\ell^\pi) := \forall {\ell}'^\pi. \, \ell^\pi \leq {\ell}'^\pi
\]
indicates that $\ell^\pi$ is the first position in $\pi$, and the formula $\last(m^\pi)$, indicating that $m^\pi$ is the last position in $\pi$, can be defined analogously.
The formula $\phi_{begins}(\pi, \pi_1)$ states that paths $\pi$ and $\pi_1$ have the same initial point,
while $\phi_{ends}(\pi, \pi_1)$ states that paths have the same final point. All of these are easily expressible in $\mwl$.

%This query, on the other hand, is expressible in \mwl. To derive this result more easily, we show that a modified query $Q^{\difm}_{\lenm}$ is expressible in \mwl. To formally define $Q^{\difm}_{\lenm}$, we first introduce the notion of the color of a node, which is definable in \wl.
As a result, we have that \mwl is more powerful than \wl in expressiveness.
%\michael{I guess it should read ``some red node to some black node''?}
%\dd{Revised as suggested.}
%It is shown that $Q^{\difm}_{\lenm}$ is not expressible in \wl \cite{Hellings:13:ICDT}. However, it is expressible in \mwl as follows:
%and
\end{proof}

In Theorem~\ref{thm:existingcontainments}, we claimed that \wl is subsumed by \edn, which follows directly from the following result:

\begin{proposition}
\mwl is subsumed by \edn.
\end{proposition}
\begin{proof}
We show that \mwl formulas can be inductively translated into\linebreak \edn formulas.
First, consider \mwl atoms:
%\begin{enumerate}[(1)]
\begin{inparaenum}[(1)]
\item\label{atom:E_a} $E_a(x,y)$.
\item\label{atom:order} $x < y$.
\item\label{atom:eqid} $x \eqid y$.
\item\label{atom:eqdata} $x \eqdata y$.
\end{inparaenum}
%\end{enumerate}
%michael: avoid bogus succinctness of "respectively"!
Consider two paths $\rho$ and $\tau$ and two node positions $i^\rho$ and $j^\tau$, Each of (\ref{atom:E_a}) to (\ref{atom:eqdata}) corresponds to a property of prefixes $\rho'$ of $\rho$ (for $x$) with length $i^\rho$ and $\tau'$ of $\tau$ (for $y$) with length $j^\tau$, as follows:
\begin{enumerate}[$(1')$]
\item\label{mod:E_a} $|\tau'| = |\rho'| + 1$ and the last symbol from $\Sigma$ along $\tau'$ is $a$.
\item\label{mod:order} The length of $\rho'$ is below  the length of $\tau'$.
\item\label{mod:eqid} The last node along $\rho'$ is the same as the last node along $\tau'$.
\item\label{mod:eqdata} The last node's data value (excluding the id)  along $\rho'$ is equal 
%\michael{referee says: should be equal not equivalent} \dd{Modified} 
to the last node data value (excluding the id) along $\tau'$.
\end{enumerate}
Each of these relations ($\ref{mod:E_a}'$) through ($\ref{mod:eqdata}'$) can be defined by \ra expressions. Note that in the case of (\ref{atom:E_a}), variables $x$ and $y$ should belong to the same sort.

Note that in \edn, there are no position variables. So above, when we translate an \mwl formula into an \edn one, we map a  position variable $x$ associated to the path variable $\pi$ to a path variable that is constrained to be a prefix of $\pi$. The property that path $x$ is a prefix of $\pi$ can be expressed by the atomic \edn formula $(x,\pi) \in A_{\textsc{prefix}}$, where $A_{\textsc{prefix}}$ is the \ra recognizing the pair $(\rho,\tau)$ of data paths where $\rho$ is a prefix of $\tau$. By using the {\ra}s defined by  ($\ref{mod:E_a}'$)-($\ref{mod:eqdata}'$), we can then recursively translate \mwl formulas into their corresponding \edn forms.
\end{proof}

On the other hand, \mwl remains fundamentally a first-order language, unable to count the parity of a set. More specifically, we show that on $\Sigma$-words, \mwl is subsumed by $\textsc{FO}[<, +, \{a\}_{a \in \Sigma}]$, a first-order logic where variables range over positions of words and the interpretations of $<, +, \{a\}_{a \in \Sigma}$ are standard. Since checking whether the number of positions with a particular symbol $a$ is even is not definable in $\textsc{FO}[<, +, \{a\}_{a \in \Sigma}]$ \cite{Libkin:04:text}, it is also not definable in \mwl.

\begin{proposition} 
Over words (i.e. chain graphs) {\mwl} is subsumed by  \\
$\textsc{FO}[<,+,\{a\}_{a\in \Sigma}]$.
\end{proposition}
\begin{proof}
Given a chain data graph $G$, we use $w_G$ to denote the label of the chain. 

\begin{claim} \label{clm:mwltofo} For each \mwl sentence $\phi$, there is an $\textsc{FO}[<,+,\{a\}_{a\in \Sigma}]$ sentence $\psi$ such that for every chain graph $G$  with only ids as data, $G \models \phi$ if and only if $w_G \models \psi$.
\end{claim}

Each path variable $\pi$ in $\phi$ is translated to two position variables $\pi^{\beginm}$ and $\pi^{\textbf{end}}$ in $\psi$, representing the endpoints. Each position variable $\ell$ in $\phi$ is translated to a position variable $x_\ell$ in $\psi$. Additionally, we enforce the following conditions:
\begin{itemize}
\item $\pi^{\beginm} \leq \pi^{\textbf{end}}$, and the equality holds if and only if $\pi$ is empty.
\item $0 \leq x_\ell < \pi^{\textbf{end}}-\pi^{\beginm}$ if $\ell$ is of sort $\pi$, and $\pi$ is not empty.
\end{itemize}

Then, we can translate \mwl atomic formulas as follows:
\begin{itemize}
\item $E_a(\ell^\pi,m^\pi)$ is translated to $a(x_\ell+\pi^{\beginm})\wedge (x_m=x_\ell+1)$.
\item $\ell < m$ is translated to $x_\ell < x_m$.
\item $\ell^\pi \eqid m^\omega$ is translated to $\ell + \pi^{\beginm} = m + \omega^{\beginm}$.
\end{itemize}

$\textsc{FO}[<,+,\{a\}_{a\in \Sigma}]$ is a first-order logic that incorporates existential quantification, disjunction, and negation, so the induction step for the corresponding \mwl operators is straightforward.
This completes the proof of Claim \ref{clm:mwltofo}.
\end{proof}

Parity is a regular language, so it is in \rd and \gpc. Thus:

\begin{proposition}
\mwl does not subsume \gpc or {\rd}.
\end{proposition}

As a result, we have:

\begin{proposition}
\mwl is strictly subsumed by \edn.
\end{proposition}

\section{Discussion} \label{sec:conc}

In this work, we have provided what we believe is a fairly complete picture of the landscape of expressiveness among query languages for data graphs. See Figure~\ref{fig:summary}. We showed that the landscape of prior query languages lacked a single maximally expressive language. \fotcd, a first-order language augmented with transitive closure, and \edn were maximally expressive query languages, but they are incomparable. To remedy this, we introduce \ednt, which is \edn extended with transitive closure. We show that it subsumes the prior languages.

Beyond exploring the landscape of query languages for data graphs, we examine the potential of \wl, which only allows comparisons along a single path. We show that \mwl, the extension of \wl allowing multi-path comparisons, is more powerful  than \wl.

Although query evaluation of \ednt is decidable, we acknowledge that it is, in an important sense, too expressive, as the worst-case complexity of evaluation is extremely high, with no elementary bound in the size of the data. Identifying fragments with better evaluation complexity, as well as developing algorithmic techniques for implementing them, is a key piece of ongoing work.

Since \rd and its extensions, including \ednt, manipulate data along the paths through register automata, we can easily extend these languages to include aggregates by replacing register automata with other automata  such as register automata with linear arithmetic \cite{reglin} and cost register automata \cite{Alur:13:LICS}. This is another component of future work.

\subsection{Other aspects of graph query languages}
A number of aspects of graph querying have \emph{not} been considered here.

\myparagraph{Property graphs vs our data model}
First, in our work we deal with a simplified data model.
A commonly considered data model for graphs with data in recent years is the \emph{property graph model} \cite{GPC:22:Francis}, and this data model is the basis for work on standards. In a property graph, both nodes and edges can have associated data. This data can  be left unrestricted, including varying the number of attributes as well as their datatypes.  Or it can be restricted by a schema. In our simplified data model, we require a fixed set of attributes for nodes, each with standard infinite datatypes, while we have a fixed set of edge relations, not associated with data values.
While we believe that our results can be extended to the property graph model (e.g. by standard encoding of edge labels within node labels), we do not formalize this.

\myparagraph{GPC} In this work, we have not considered several aspects of GPC, including  inverse edges and trails. 

\myparagraph{Inverses} As mentioned earlier, inverse edges are not part of the definition of many of the classical graph query languages (including R(D)PQ and Walk Logic) and a sensible comparison would need to either include inverse edges in these languages, or consider the fragment of GPC without inverse edges. We have chosen the latter, to simplify our presentation. The former would also not change the subsumption of the query languages, and take us to languages like 2RPQ \cite{CGLV00} and extension with data \cite{FJL22}.

\myparagraph{Trails} We now discuss trails, which restrict paths so that no edges are repeated. We do not consider trails in detail here. But we note that including trails does not change the subsumption of GPC in the other query languages that we considered in this paper. For example, we can assert that a path is a trail in Walk Logic as follows:
\begin{eqnarray*} 
\phi_{trail}(\pi) & := & \forall \ell^{\pi},m^{\pi}.\ \ell^{\pi} \neq m^{\pi} \wedge \ell^{\pi} \eqid  m^{\pi} \wedge \\
    &  & \bigvee_{a \in \Sigma} (E_a(\ell^{\pi},\ell^{\pi}+1) \wedge
E_a(m^{\pi},m^{\pi}+1))
\rightarrow
(\ell^{\pi}+1 \not\eqid  m^{\pi}+1).
\end{eqnarray*}
Similarly, we can express trails in
\edn by using similar register automata as in Example \ref{ex:edn_Hamiltonian}, but with two registers instead of just one.

%For instance, in Figure~\ref{fig:DataGraph_ex}, a node represents a person. At first glance, this might suggest that our model cannot represent entities of different kinds. In fact, we can transform a property graph into a data graph and vice versa in linear time. Specifically, in our model, node labelling can be achieved by introducing a node $n$ such that, for all other nodes $n'$ in the original graph $G$, there is an edge from $n$ to $n'$ labelled with the desired label. 

%Using this technique, we can also convert any data graph $G$ with $k \geq 2$ into a graph $G'$ where $k = 1$. For example, consider the data graph $G$ in Figure~\ref{fig:DataGraph_ex}. To handle a data attribute such as ``Age,'' we introduce a new edge label symbol. We then construct a new graph $G'$ where:
%\begin{enumerate}
%    \item $V'$ is a disjoint copy of the nodes in $V$ of $G$, and  
%    \item the set of nodes in $G'$ becomes $V \cup V'$.
%\end{enumerate} 
%For each node $n \in V$, we add an edge from $n$ to its corresponding node $n' \in V'$ with the label ``Age.'' 
%From now on, \emph{we assume $k = 1$ by default}. All results hold for $k > 1$ with the same proofs. The assumption $k > 0$ will be important only in a few of our separation results: 
%see properties P.\ref{prop:sublanguage_incomparable} and P.\ref{prop:edn_not_contain_gxpath} in the proof of Theorem \ref{thm:prior}.

\myparagraph{Joins}
We have discussed \gpc \emph{pattern queries} here, which return paths. It is possible to extend these queries to return more complex objects, particularly tuple of paths. For two \gpc queries $Q$ and $Q'$ with semantics $\semantics{Q}{G}$ and $\semantics{Q'}{G}$, respectively, we define their \emph{join} as follows:
\begin{itemize}
    \item $\semantics{Q, Q'}{G} = \{((\bar{\rho}, \bar{\rho}'), \mu \cup \mu') \mid (\bar{\rho}, \mu) \in \semantics{Q}{G}, (\bar{\rho}', \mu') \in \semantics{Q'}{G}$, and the union $ \mu\cup\mu' \text{ is defined}\}$.
\end{itemize}

Thus a join of two path patterns returns a pair of paths.
We can close the syntax of \gpc under join, allowing us to return $k$-tuples for any $k$.
In our paper we compare languages by considering the graph queries that return Booleans which they can express (Definition \ref{def:subsumes} in Section \ref{sec:preliminaries}).
Thus, since we do not consider queries that return tuples, and we cannot use joins to produce new paths or Booleans, the presence or absence of joins has no impact on our results.

%\end{definition}




\bibliographystyle{elsarticle-num} 
\bibliography{refs}

\appendix
%\newpage
%\section{Appendix}

%\input{sec_hyperproperties}

%\input{sec_iltl}

\end{document}